\documentclass{elsarticle}
\usepackage{graphics}
\usepackage{color}
\usepackage[utf8x]{inputenc}
\usepackage{amssymb,amsbsy,amsmath,amsfonts,amssymb,amscd,amsthm}
\newcommand{\NP}{\mathrm{NP}}     
\newcommand{\enum}[1]{\textsc{Enum}\smash{\cdot}#1}

\newdefinition{definition}{Definition}
\newdefinition{example}{Example}
\newdefinition{remark}{Remark}
\newtheorem{proposition}{Proposition}
\newtheorem{theorem}{Theorem}
\newtheorem{lemma}{Lemma}
\newtheorem{corollary}{Corollary}
\DeclareMathOperator{\rank}{rank}

\title{Monadic second-order model-checking on decomposable matroids}
\author{Yann Strozecki}
\ead{strozecki@logique.jussieu.fr}
\address{\'Equipe de logique math\'ematique, Universit\'e Paris 7 - Denis Diderot}

\begin{document}

\begin{abstract}
A notion of branch-width, which generalizes the one known for graphs, can be defined for matroids.
 We first give a proof of the polynomial time model-checking of monadic second-order formulas on representable matroids of bounded branch-width, 
 by reduction to monadic second-order formulas on trees. This proof is much simpler than the one previously known.
 We also provide a link between our logical approach and a grammar that allows to build matroids of bounded branch-width.
Finally, we introduce a new class of non-necessarily representable matroids, described by a grammar and
on which monadic second-order formulas can be checked in linear time. 
\end{abstract}

\maketitle
\section{Introduction}

The model-checking of monadic second-order formulas is a natural and extensively studied problem
that is relevant to many fields of computer science such as verification or database theory. 
This problem is hard in general (since monadic second-order logic, $MSO$ for short, can express $\NP$-complete properties like $3$-colorability) 
but it has been proved tractable on various structures.
For example, it is  decidable in linear time on trees \cite{thatcher1968gfa}
 thanks to automata techniques. It also remains linear time decidable \cite{courcelle1991graph} on the widely studied class of graphs 
of bounded tree-width. Since then, a lot of similar results  have been found, either with similar
notions of  width, like clique-width and rank-width, or for extensions of $MSO$, for instance by counting predicates (see \cite{Logic_graph_algorithm,hlineny2008width}). 

In this article, we study the model-checking of monadic second order sentences on matroids 
and especially on representable matroids, which are a natural generalization of both graphs and matrices. 
Natural notions of decomposition such as tree-width or branch-width can be adapted in this context.
It is also interesting to note that tree-width and branch-width on matroids are generalizations of the same notions
on graphs. In fact, the branch-width of a $2$-connected graph is equal 
to the branch-width of its cycle matroid \cite{hicks2007branchwidth,mazoit2007branchwidth}.

 The monadic second-order logic on matroids, 
denoted by $MSO_M$, enables to express many interesting matroids properties (see \cite{hlineny2003matroid} and the references therein)
such as the connectivity or the representability over $\mathbb{F}_2$ or $\mathbb{F}_3$.
Recently, the model-checking of $MSO$ formulas on representable matroids of bounded branch-width has been studied and
 it has been proved to be decidable in a time linear in the size of the matroid \cite{hlinveny2006bwp}.
 This result has been subsequently extended in \cite{kral2010decomposition} to a broader class of matroids.
The first contribution of this article is to introduce an alternative method to study these matroids, by an appropriate decomposition 
into labeled trees, called \emph{enhanced trees} and a translation of $MSO_{M}$ into $MSO$. For this purpose, we introduce the notion of \textit{signature} 
over decomposed matroids which appears to be a useful general tool to study several classes of matroids. 
Signatures can be seen as the states of a nondeterministic bottom-up automaton which checks the dependence of 
a set of a matroid represented by a set of leaves of an enhanced tree.

As a corollary of this method, we give a new proof of the linear time model-checking of $MSO_{M}$ formulas on representable matroids of bounded branch-width,
 and also an enumeration algorithm of all tuples satisfying a $MSO_{M}$ query with a linear delay.
We apply this result to the problem $\textsc{A-Circuit}$, which asks to extend a set of a matroid into a circuit.
Both decision and enumeration versions of this problem have been well studied \cite{1113216} and, in the case of
$\mathbb{F}$-matroids of bounded branch-width, we obtain better algorithms.

From  this starting result,  we  derive a general way to build matroid grammars, inspired by the parse tree of \cite{hlinveny2006bwp}. 
We first introduce a grammar for matrices, which is similar to the one for representable matroids introduced in \cite{hlinveny2006bwp}.
We show why it is more appropriate to see this grammar as a matrix one rather than a matroid one. Thanks to the connection with enhanced trees,
 we easily prove that it describes  
the representable matroids of bounded branch-width. We then build the class of matroid $\mathcal{T}_k$ by means of series-parallel operations. 
Since these operations allows to combine non representable matroids, the class $\mathcal{T}_k$ contains different matroids
than those studied in the first part of the article.
 We give some useful insights about the structures of matroids in $\mathcal{T}_k$ and its relations with the branch-width.
Using the same approach as for the matroids of bounded branch-width, we build a $MSO$ formula expressing the dependence relation over terms 
representing a matroid of $\mathcal{T}_k$. It enables us to prove that the model-checking of $MSO_M$ formulas is decidable in linear time over $\mathcal{T}_k$.

\section{Matroids and Branch-width}
\subsection{Matroids}
\label{sec:matroid}
Matroids have been designed to abstract the notion of dependence that appears, for example, in graph theory or in linear algebra.
All needed informations about matroids (and the proofs of what is stated in this section) can be found in the book Matroid Theory by J. Oxley \cite{oxley1992mt}.
\begin{definition}
A matroid\index{matroid} is a pair $(S,\mathcal{I})$ where $S$ is a finite set, called the ground set, and
 $\mathcal{I}$ is included in $\mathcal{P}(S)$, the power set of $S$.
Elements of $\mathcal{I}$ are said to be independent sets, the others are dependent sets.
A matroid must satisfy the following axioms:

\begin{enumerate}
\item $\emptyset \in \mathcal{I}$
 \item If $I \in  \mathcal{I}$ and $I' \subseteq I$, then $I'  \in \mathcal{I}$
 \item If $I_{1}$ and $I_{2}$ are in $\mathcal{I}$ and $|I_{1}| < |I_{2}|$, then
there is an element $e$ of $I_{2} - I_{1}$ such that $I_{1} \cup e \in  \mathcal{I}$.
\end{enumerate}
\end{definition}

The matroids, like the graphs, may have \emph{loops}, which are dependent singletons,
but in all this article we assume that the matroids are loop free.
In a matroid,  a \textit{base} is a maximal independent set for inclusion.
A \textit{circuit}\index{circuit} is a minimal dependent set for inclusion.

Let $M$ be a matroid, $S$ a subset of its elements, the restriction\index{restriction} of $M$ to $S$,
written $M|S$ is the matroid $(S,\mathcal{I})$, such that a set is in $\mathcal{I}$ if it is independent in $M$ and contained in $S$.

A function $f$ from a matroid $M$ to a matroid $N$, is a morphism of matroids if, 
for all dependent set $S$ of $M$, $f(S)$ is dependent in $N$. An isomorphism 
is a bijection such that itself and its inverse are matroid morphisms.

We can represent any finite matroid by giving the collection of its independent sets,
which can be exponential in the size of the ground set.
One usual way to address this problem is to assume that the matroid is represented by a black box deciding in unit time 
if a set is independent or not, see \cite{1113216}.
We also consider subclasses of matroids, for which we do not need the explicit set of independent sets,
because we can decide if a set is independent or not in polynomial time. The two following examples,
and the classes of matroids introduced in Sec.\ref{sec:terms}, are of this nature. 

\paragraph{Vector Matroid}

Let $A$ be a matrix, the \emph{vector} matroid of $A$
has for ground set the columns of $A$ and a set of column vectors is independent
if it is linearly independent.

\begin{definition}
 A matroid $M$ is representable over the field $\mathbb{F}$ if it is isomorphic to a vector matroid of a matrix $A$ with coefficients in $\mathbb{F}$.
We also say that $M$ is represented by $A$ and that $M$ is a $\mathbb{F}$-matroid.
\end{definition}

Note that there are matrices which are not similar\footnote{a matrix $A$ is similar to $B$ if there is an invertible matrix $S$ such that
 $A=SBS^{-1}$} but represent the same matroid. The matroids representable over $\mathbb{F}_{2}$ are called \emph{binary} matroids and those which are representable over any field
are called \emph{regular} matroids.

\begin{example}
 \begin{displaymath}
A =
\left( \begin{array}{ccccc}
1 & 0 & 1 & 0& 1 \\
1 & 1 & 0 & 0& 1  \\
0 & 1 & 1 & 1& 1 
\end{array} \right)
\end{displaymath}
The matrix $A$ is defined over $\mathbb{F}_2$.
The convention is to name a column vector by its position in the matrix.
Here the set $\left\lbrace 1,2,4 \right\rbrace $ is independent while  $\left\lbrace 1,2,3 \right\rbrace $ is dependent.
\end{example}

\paragraph{Cycle Matroid}
The second example is the \emph{cycle} matroid\index{cycle matroid} of a graph; such matroids are said to be \emph{graphic}.
Let $G$ be a graph, the ground set of its cycle matroid is the set of its edges.
A set is said to be dependent if it contains a cycle.
Here a base is a spanning tree if the graph is connected and a circuit is a cycle.

\begin{example}
\begin{figure}[ht]
\centering
\begin{tabular}{c c}
\parbox[l]{6cm}{ 
\begingroup
  \makeatletter
  \providecommand\color[2][]{%
    \errmessage{(Inkscape) Color is used for the text in Inkscape, but the package 'color.sty' is not loaded}
    \renewcommand\color[2][]{}%
  }
  \providecommand\transparent[1]{%
    \errmessage{(Inkscape) Transparency is used (non-zero) for the text in Inkscape, but the package 'transparent.sty' is not loaded}
    \renewcommand\transparent[1]{}%
  }
  \providecommand\rotatebox[2]{#2}
  \ifx\svgwidth\undefined
    \setlength{\unitlength}{175.99999422pt}
  \else
    \setlength{\unitlength}{\svgwidth}
  \fi
  \global\let\svgwidth\undefined
  \makeatother
  \begin{picture}(1,0.47384604)%
    \put(0,0){\includegraphics[width=\unitlength]{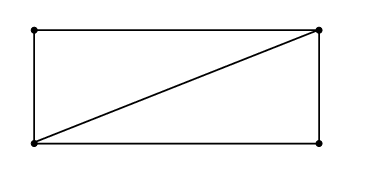}}%
    \put(0.49842942,0.43334145){\color[rgb]{0,0,0}\makebox(0,0)[lb]{\smash{$1$}}}%
    \put(-0.00442531,0.22808854){\color[rgb]{0,0,0}\makebox(0,0)[lb]{\smash{$2$}}}%
    \put(0.47898813,0.0078354){\color[rgb]{0,0,0}\makebox(0,0)[lb]{\smash{$3$}}}%
    \put(0.90261717,0.22492012){\color[rgb]{0,0,0}\makebox(0,0)[lb]{\smash{$4$}}}%
    \put(0.37334788,0.28088894){\color[rgb]{0,0,0}\makebox(0,0)[lb]{\smash{$5$}}}%
  \end{picture}%
\endgroup
}
&
$
X =
\left( \begin{array}{ccccc}
1 & 0 & 0 & 1 & 1 \\
1 & 1 & 0 & 0 & 0  \\
0 & 1 & 1 & 0 & 1 \\
0 & 0 & 1 & 1 & 0
\end{array} \right)
$
\end{tabular}
\caption{A graph and its representation by a matrix over $\mathbb{F}_2$}
 \label{fig:graph}
\end{figure}
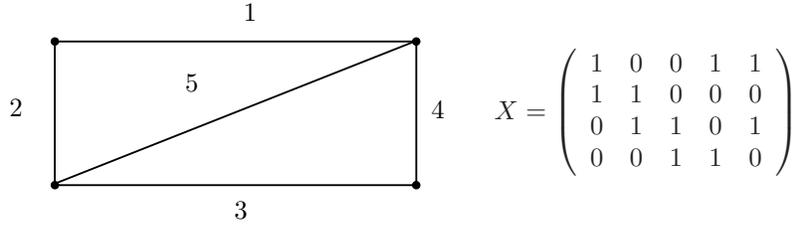
In Fig. \ref{fig:graph}, the set $\left\lbrace 1,2,4 \right\rbrace $ is independent whereas $\left\lbrace 1,2,3,4 \right\rbrace $ and $\left\lbrace 1,2,5 \right\rbrace $ are dependent.
\end{example}

\begin{remark}
Any cycle matroid of a graph $G$ is a binary matroid.  
To see this, one chooses an order on the edges and on the vertices of $G$ then build the incidence matrix of $G$ over any field.
The dependence relation is the same over the edges and over the vectors representing the edges.
\end{remark}

\paragraph{Other representations}

One interest of matroids is that they are axiomatizable in a number of different ways.
The rank\index{rank} function on matroids, similar to the rank function 
on vector spaces, plays an important role.
It is defined by:
$$\rank(B) = \max \{ |A| \, \mid A \subseteq B \text{ and } A \text{ independent}\}$$ 
This function is monotonic, that is for all subsets $X$ and $Y$, $X \subseteq Y \Rightarrow r(X) \leq r(Y)$.
It is also submodular, that is for all $A$ and $B$: 
$$\rank(A \cup B) + \rank(A \cap B) \leq \rank(A) + \rank(B) $$
and it even leads to a characterization of matroids:
\begin{proposition}
Let $S$ be a finite set and $r$ a function from $\mathcal{P}(S)$ to $\mathbb{Z}$. The function $r$ is the rank of a matroid if and only if it is submodular,
 monotonic and such that the rank of any element is $0$ or $1$.  
\end{proposition}

One can also define a matroid by the collection of its circuits.

\begin{proposition} Let $S$ be a finite set and let $\mathcal{C}$ be a subset of $\mathcal{P}(S)$.
  Then $\mathcal{C}$ is the set of circuits of a matroid if and only if it satisfies the following axioms:
\begin{enumerate}
\item $\emptyset \notin \mathcal{C}$
\item If $C_{1},C_2 \in \mathcal{C}$ and $ C_{1} \subseteq C_{2}$ then $C_{1}= C_2$
 \item If $C_{1},C_{2} \in \mathcal{C}$, $C_1 \neq C_2$ and $e \in C_{1} \cap C_{2}$ then there is $ C_3 \in \mathcal{C}$
 such that $C_3 \subseteq C_{1} \cup C_{2} \setminus\left\lbrace e \right\rbrace $ 
\end{enumerate}
\end{proposition}

\subsection{Branch Decomposition}

In this subsection we define the branch-width of a matroid, thanks to the more general notion of connectivity function,
which also allows us to define the branch-width of a graph.
We follow the presentation of \cite{Logic_graph_algorithm}.

Let $S$ be a finite set and $\kappa: 2^{S} \rightarrow \mathbb{N}$.
The function $\kappa$ is symmetric if $\kappa(B) = \kappa(S \setminus B)$ for all $B \subseteq S$.
If $\kappa$ is symmetric and submodular, it is a \textit{connectivity function}.

A branch decomposition of $(S,\kappa)$ is a pair $(T,l)$ where $T$ is a binary tree and $l$ is a one to one labeling of the leaves
of $T$ by the elements of $S$.
We define the mapping $\tilde{l}$, from the vertices of the graph to the sets of $S$ recursively:

 $$\tilde{l}(t) = \left\lbrace  \begin{array}{l l}
 \left\lbrace l(t)\right\rbrace  &     \text{ if } t \text{ is a leaf}\\
 \tilde{l}(t_{1}) \cup \tilde{l}(t_{2}) & \text{ if } t \text{ is an inner node with children } t_{1}, t_{2}\\
  \end{array}\right.$$

The width of the branch decomposition of $(S, \kappa)$ is defined by
$$width(T, \tilde{l})= \max \left\lbrace \kappa(\tilde{l}(t)) \mid t \in V(T) \right\rbrace $$

The branch-width of $(S,\kappa)$ is the minimum of the width over all branch decompositions.
Thanks to a result of Iwata, Fleischer and Fujishige \cite{iwata2001csp} about minimalization of a submodular function, 
we know that we can find an almost optimal branch decomposition with a fixed parameter tractable algorithm.

\begin{theorem}[Oum and Seymour \cite{oum2006acw}]
\label{oum}
For any given $k$, there is an algorithm as follows. It takes as input a finite set 
$S$ and $\kappa$ a polynomial time computable connectivity function such that $\kappa(\{v\}) \leq 1$.
The algorithm concludes in polynomial-time either that $bw(S,\kappa) > k$
or outputs a branch decomposition of $(S,\kappa)$ of width at most $3k+1$.
\end{theorem}

We now define a connectivity function adapted to the matroid case.
 Let $M$ be a finite matroid with ground set $S$ and let $B$ be a set of elements of $M$.
We define the connectivity function by $\kappa(B) = \rank(B) + \rank(S \setminus B) - \rank(S)$.
The function $\kappa$ is symmetric by construction and submodular because the rank function is submodular.

In this article, we restrict our study of branch-width to representable matroids.
It means that $M$ is given as a matrix $A$ over a field $\mathbb{F}$.
In this case, the rank relative to matroids is equal to the rank in the sense of linear algebra.
Moreover, the rank of a family of column vectors $B$ is the dimension of the vector subspace it generates, denoted by $<B>$.

The following holds, where $+$ is the sum of vector spaces:
 $$\dim(<S>)= \dim(<B> + <S \setminus B>).$$
Therefore, by a classical theorem on the dimension of the sum of two vector
spaces, we obtain :
$$\dim(<S>) = \dim(<B>) + \dim(<S \setminus B>) - \dim (<B> \cap <S \setminus B>).$$
We replace $\dim(<S>)$ by this expression in the definition of $\kappa$ to obtain:
$$\kappa(B) = \dim (<B> \cap <S \setminus B>).$$

Let $(T,l)$ be a branch decomposition of width $t$ of $M$ and let $s$ be a node of $T$.
In this article, we note $T_{s}$ the subtree of $T$ rooted in $s$ and $E_{s}$ the vector subspace generated by $\tilde{l}(s)$, that is to say the set of leaves of $T_{s}$.
  Let $E_{s}^{c}$ be the subspace generated by $S  \setminus \tilde{l}(s)$ i.e. the set of leaves which do not belong to $T_{s}$.
Let $B_{s}$ be the subspace $E_{s} \cap E_{s}^{c}$, it is the boundary between what is described inside and outside of $T_{s}$.

\begin{remark}
We have seen that $\kappa(\tilde{l}(s)) = \dim(E_{s} \cap E_{s}^{c})$ which is equal by definition to $\dim(B_{s})$.
If $t$ is the width of the branch decomposition $(T,l)$, for all nodes $s$ of $T$, $ \dim B_{s} \leq t $.
\label{dimension}
\end{remark}

\begin{figure}
\begin{flushleft}
\begin{tabular}{c c c}
$
X =
\left( \begin{array}{cccccc}
1 & 1 & 0 & 0 & 1 & 1\\
0 & 1 & 0 & 1 & 1 & 0\\
1 & 1 & 1 & 0 & 0 & 0\\
0 & 1 & 0 & 0 & 0 & 0
\end{array} \right)
$&
\parbox[l]{4cm}{\includegraphics[scale=0.3]{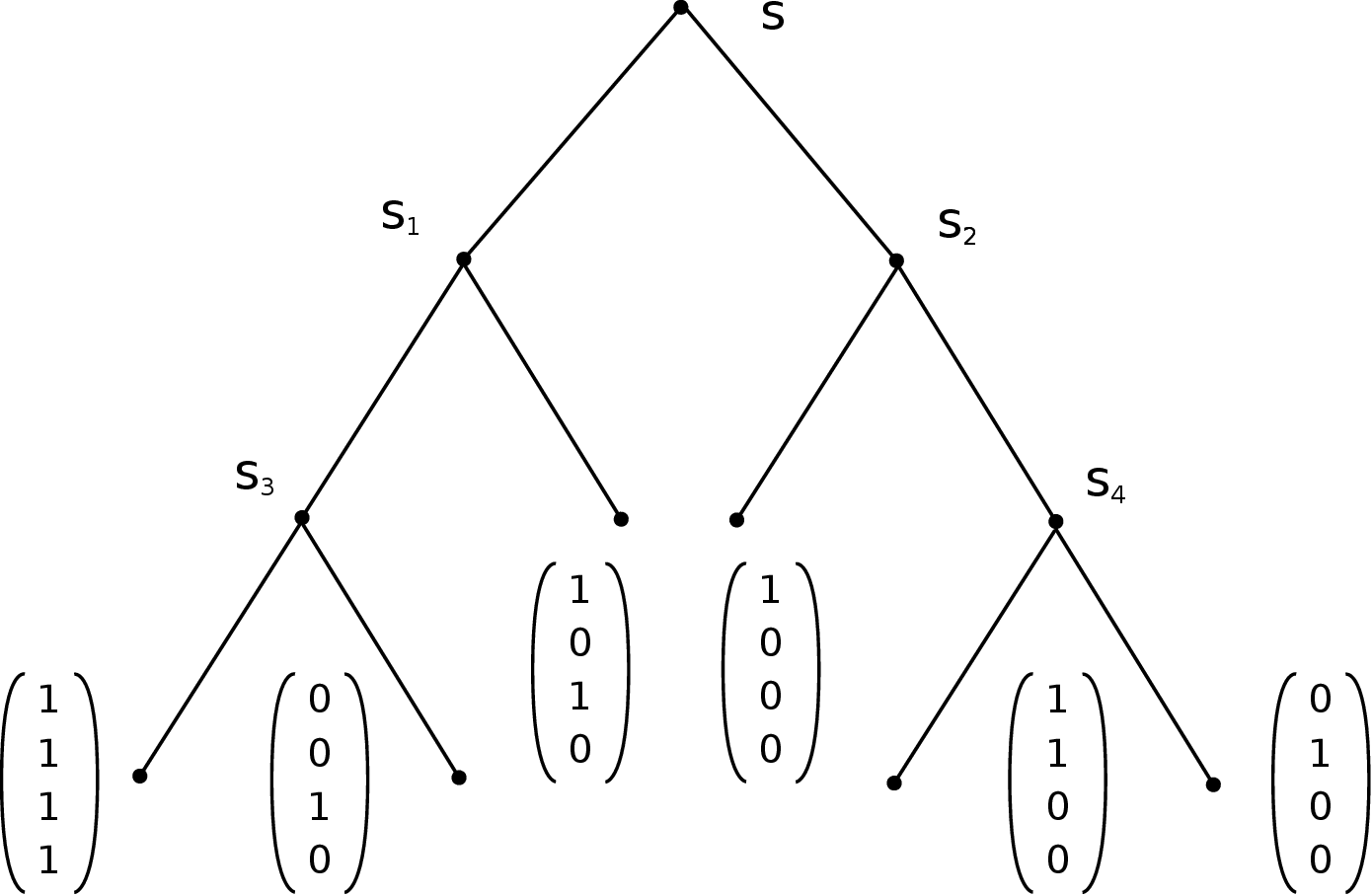}}
\end{tabular}
\caption{A matrix $X$ and one of its branch decomposition of width $1$}
\label{fig:tree}
\end{flushleft}
\end{figure}

\begin{example}
 To illustrate this notion, we compute $E_{s_{1}}$ and $E_{s_{1}}^{c}$ to find $B_{s_{1}}$ in the tree of
 Fig. \ref{fig:tree}. Notice that, when $s$ is a leaf, the subspace $E_{s}$ is generated by the single vector $l(s)$.
Therefore $B_s = E_{s} \cap E_{s}^{c}$ is either equal to $E_{s}$ or trivial, i.e. equal to the zero vector, as in the case of the left child of $s_{3}$
in Fig. \ref{fig:tree}.
 \medskip
 
 \noindent
 \begin{footnotesize}
$E_{s_{1}} = < \left(\begin{array}{c}1\\1\\1\\1\end{array}\right),\left(\begin{array}{c}0\\0\\1\\0\end{array}\right),\left( \begin{array}{c}1\\0\\1\\0\end{array}\right)>$; 
$E_{s_{1}}^{c} = <\left(\begin{array}{c}1\\0\\0\\0\end{array}\right),\left(\begin{array}{c}0\\1\\0\\0\end{array}\right)>$;  $B_{s_{1}}= <\left(\begin{array}{c}1\\0\\0\\0\end{array}\right)>$ 
\end{footnotesize}
\end{example}

In the case of a matroid representable on a finite field, we have the following result, similar to Theorem \ref{oum},
which gives an exact decomposition algorithm.

\begin{theorem}[Hlin\v{e}n\'y and Oum \cite{hlineny2008finding}]
  Let $k$ be a fixed integer and let $\mathbb{F}$ be a fixed finite field.
There is an algorithm which given as input an $\mathbb{F}$-matroid $M$ outputs in cubic time (parametrized by $k$ and $|\mathbb{F}|$),
a branch-decomposition of $M$ of width at most $k$, or confirms that $bw(M) > k$.
\end{theorem}

\section{Enhanced Branch Decomposition Tree}

From now on, all matroids will be representable over a fixed finite field $\mathbb{F}$.
The results of the next part are false if $\mathbb{F}$ is not finite, see \cite{hlinveny2006bwp}.
As we do not know how to decide in polynomial time if a matroid is representable, when we say it is,
we assume that it has been given as a matrix. Furthermore, to simplify the presentation,
we assume that the matroids have no loops, but this condition could easily be lifted.

Let $t$ be a fixed parameter representing the maximal branch-width of the considered matroids.
Let $M$ be a matroid represented by the matrix $A$ over $\mathbb{F}$ and $(T,l)$ one of its branch decomposition 
 of width at most $t$. We will sometimes not distinguish a leaf $v$ of $T$ from the column vector $l(v)$ it represents.
 Let $E$ be the vector space generated by the column vectors of $A$, we suppose that its dimension is the same 
as the length of the columns of $A$ and we denote it by $n$. 

We now build, for each node $s$, a matrix $C_s$. The construction is by bottom-up induction, that is from leaves to root. 
The column vectors of this matrix are elements of $E$ and they are partitioned in three parts
which are bases of subspaces of $E$.
If $s$ is a leaf, $C_{s}$ is a base vector of the subspace $B_{s}$.
If $s$ has two children $s_{1}$ and $s_{2}$, the matrix $C_{s}$ is divided in three parts $\left( C_{1}| C_{2}| C_{3} \right) $
where $C_{1}$, $C_{2}$ and $C_{3}$ are bases of $B_{s_{1}}$, $B_{s_{2}}$ and $B_{s}$ respectively.
By induction hypothesis, one already knows the bases of $B_{s_{1}}$ and $B_{s_{2}}$ used to build $C_{s_{1}}$ 
and $C_{s_{2}}$ and we choose them for $C_{1}$ and $C_{2}$. We then choose any base of $B_{s}$ for $C_{3}$.

Matrices $C_{s}$ are of dimension $n\times t_{1}$, with $t_{1} \leq 3t$, because of Remark \ref{dimension} on the dimension of boundary subspaces. 
A \emph{characteristic matrix} at $s$ is obtained by selecting a maximal independent set of rows of $C_{s}$ by Gaussian elimination.
The result is a $t_{2}\times t_{1}$ matrix $N_{s} = \left(N_{1}| N_{2}|N_{3}\right)$ with $t_{2} \leq 3t$. 

The vectors in $N_{s}$ still represent the bases of  $B_{s_{1}}$, $B_{s_{2}}$ and $B_{s}$ in the same order but they only carry the dependence information. 
In fact, any linear dependence relation between the columns of $C_{s}$ is a linear dependence relation between the same columns of $N_{s}$ with the same coefficients, and conversely.  
Note that the characteristic matrix at a node is not unique. It depends on the choice of bases used to represent the subspaces $B_{s}$ 
and on the rows which have been removed by Gaussian elimination.

\begin{definition}[Enhanced branch decomposition tree]
Let $M$ be a $\mathbb{F}$-matroid and let
$(T,l)$ be one of its branch decomposition of width $t$. 
Let $\tilde{T}$ be the tree $T$ labeled at each node by a characteristic matrix obtained by the previous construction.
We say that $\tilde{T}$ is an enhanced branch decomposition tree of $M$ of width $t$ (enhanced tree for short). 
\end{definition}

Each label can be represented by a word of size polynomial in $t$. This is the reason why 
the matrix $N$ has been chosen instead of $C$ which is of size linear in the matroid.
Indeed, the labels later appear in a formula whose size must depend only in $t$.
Note also that the leaves of the enhanced tree are in bijection with the elements of the matroid,
by the same function $l$ as for the branch decomposition tree.

Remark that, given a $n\times m$ matrix and a branch decomposition tree of width $t$, one can transform this tree into an enhanced tree in cubic time.
The transformation of the matrix $C$ into $N$ only takes a linear time, since $C$ is of size at most $3tn$. 
Assume that the matrix $A$ we are working with has been given in the normal form $(I|X)$ where $I$ is the identity matrix.
If not, it is always possible to compute such a normal form in cubic time.
We must now build a matrix $C$ for each node $s$ of the tree,
that is to find a base of the boundary space $B_s$. 
To do that we must compute the intersection of the vector space spanned by some columns of $A$ with the one spanned by the other ones.
Since $A$ is in normal form and that the intersection we compute is of dimension less than $t$, it can be done in quadratic time.

\begin{example}
 Figure \ref{fig:enhanced} represents an enhanced tree constructed from the branch decomposition tree of Figure \ref{fig:tree}.
Some of the intermediate computations needed to find it are also given for illustration.
One may check that each label $N_{s}$ of the tree is obtained by Gaussian elimination from $C_{s}$.
Remark that, as it is a decomposition of branch-width $1$,
the subspaces $B_{s}$ are of dimension $1$ and are thus represented here by one vector.

\begin{figure}
 \begin{flushleft}
\begin{tabular}{c c}
 \parbox[c]{6cm}{\includegraphics[scale=0.3]{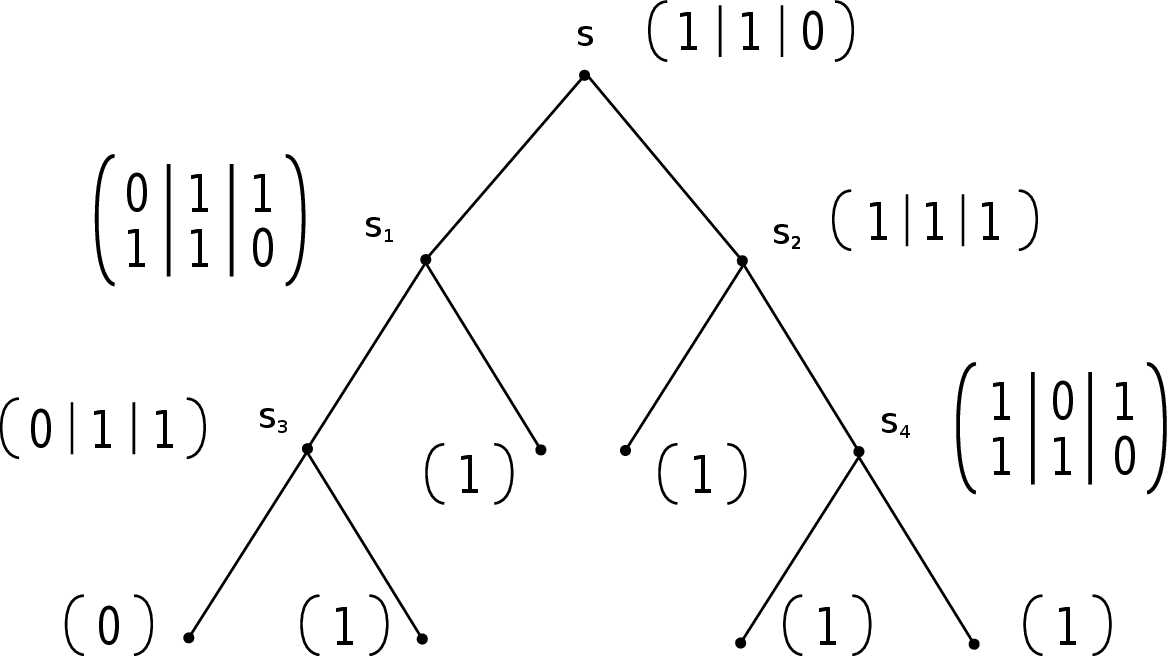}}&
$
\begin{array}{c c}
B_{s_{3}} = <\left( \begin{array}{c} 0\\ 0\\ 1\\0\end{array}\right)>  & C_{s_{3}}= \left( \begin{array}{c | c | c} 0&0&0\\0&0&0\\0&1&1\\0&0&0\end{array}\right)\\
\\
B_{s_{4}} = <\left( \begin{array}{c} 1\\ 0\\ 0\\0\end{array}\right)>  & C_{s_{4}}= \left( \begin{array}{c | c | c} 1&0&1\\1&1&0\\0&0&0\\0&0&0\end{array}\right)
\end{array} $
\\
\\ $B_{s_{1}} = <\left( \begin{array}{c} 1\\ 0\\ 0\\0\end{array}\right)> \, C_{s_{1}}= \left( \begin{array}{c | c| c}0&1&1\\0&0&0\\1&1&0\\0&0&0 \end{array}\right)$&
$B_{s_{2}} = <\left( \begin{array}{c} 1\\ 0\\ 0\\0\end{array}\right)>  \, C_{s_{2}}= \left( \begin{array}{c | c | c} 1&1&1\\0&0&0\\0&0&0\\0&0&0\end{array}\right)$
\end{tabular}

 \caption{An enhanced tree built from the branch decomposition tree of Figure \ref{fig:tree}.}
\label{fig:enhanced}
\end{flushleft}
\end{figure}
\end{example}

\section{Model-checking and Enhanced Tree}
\label{sec:rep_matroid}
\subsection{Signature}

We show how dependent sets of a matroid of bounded branch-width can be characterized using its enhanced tree.
This will later allow us to define the dependence predicate by a formula in $MSO$. 
To this aim, for any node $s$ of an enhanced tree, we succintly represent by a \textit{signature} each element of $B_{s}$ which can be generated by a given set
of elements of the matroid.

\begin{definition}[Signature]
A signature is a finite sequence of elements of $\mathbb{F}$, denoted by $\mathbf{\lambda} =  (\lambda_{1}, \dots, \lambda_{l}) $
or by $\varnothing$ when it is of length $0$.
\end{definition}

\begin{definition}[Signatures of a set]
 Let $A$ be a matrix representing a matroid and $T$ one of its enhanced tree.
Let $s$ be a node of $T$ and let $X$ be a subset of the leaves of $T_{s}$ which are seen as columns of $A$.
Let $v$ be an element of $B_{s}$, obtained by a nontrivial\footnote{at least one of the coefficient of the linear combination is not zero} linear combination 
of elements of $X$. Let $c_{1}, \dots, c_{l} $ denote the column vectors of the third part of $C_{s}$. They form a base of $B_{s}$.
Thus there is a signature $\lambda= (\lambda_{1}, \dots, \lambda_{l})$ such that $v = \displaystyle{\sum_{i=1}^{l} \lambda_{i} c_{i}}$.
We say that $X$ admits the \textit{signature} $\lambda$ at $s$.
The set $X$ also always admits the signature $\varnothing$ at $s$.
\end{definition}

The size $l$ of a signature at a node $s$ is the dimension of $B_s$, thus it is at most $t$, the width 
of the branch decomposition used to build $T$. Notice also that 
a set admits a lot of different signatures, in fact they form a vector subspace of $\mathbb{F}^{l}$ (without maybe the zero vector).
\begin{example}
\label{leaf}
 We illustrate the previous definition in the case of a leaf $s$.\\
Case $1$: the set $X$ is empty, therefore there is no combination of its elements
and the only signature it admits at $s$ is $\varnothing$.\\
Case $2$: the set $X= \{x\}$ and the label of $s$ is the matrix $(0)$. 
The space  $B_s$ is the zero vector and there is no nontrivial combination of $x$ equals to $0$,
 therefore the only signature $X$ admits at $s$ is $\varnothing$.\\
Case $3$:  the set $X= \{x\}$ and the label of $s$ is the matrix $(\alpha)$ with $\alpha \neq 0$. 
The set $B_s$ is hence generated by $x$ and the set $X$ admits the signatures 
 $(\lambda)$ for all $\lambda \neq 0 \in \mathbb{F}$ and $\varnothing$.
\end{example}

We now define a relation which describes how the signature a node admits is related to the signatures its children admit.
\begin{definition}\label{def:relation}
 Let $N$ be a matrix over $\mathbb{F}$ divided in three parts $\left(N_{1}| N_{2}|N_{3}\right) $, and let
$\lambda$, $\mu$, $\delta$ be three signatures over $\mathbb{F}$. The submatrix $N_i$ has $l_i$ columns,
and its $j^{\text{th}}$ vector is denoted by $N_i^{j}$.
The relation $R(N, \lambda,\mu,\delta)$ is true if: 
\begin{itemize}
 \item $\lambda=\mu=\delta= \varnothing$ or
 \item $\lambda$ and at least one of $\mu$, $\delta$ are not $\varnothing$ and the following equation holds
\begin{equation}
 \displaystyle{\sum_{i=1}^{l_1} \mu_{i} N_{1}^{i} + \sum_{j=1}^{l_2} \delta_{j} N_{2}^{j}= \sum_{k=1}^{l_3} \lambda_{k} N_{3}^{k}}
\label{eq}
\end{equation}
\noindent
If a signature is $\varnothing$, the corresponding sum in Eq. \ref{eq} is replaced by $0$.
\end{itemize}
\end{definition}

\begin{lemma}
\label{build}
Let $T$ be an enhanced tree, $s$ one of its nodes with children $s_{1}$, $s_{2}$
 and $N$ the label of $s$.
Let $X_{1}$ and $X_{2}$ be two sets of leaves chosen amongst the leaves of $T_{s_{1}}$ and $T_{s_{2}}$ respectively.
If $X_1$ admits $\mu$ at $s_{1}$, $X_{2}$ admits $\delta$ at $s_{2}$ and 
 $R(N, \lambda,\mu,\delta)$ holds then $X = X_{1} \cup X_{2}$ admits $\lambda$ at $s$. 
\end{lemma}

\begin{proof}
The case $\lambda=\mu=\delta= \varnothing$ is obvious.
By construction of $N$, we know that $N_{1}$, $N_{2}$ and $N_{3}$ represent the bases $C_{1}$, $C_{2}$ and $C_{3}$ of  $B_{s_{1}}$, $B_{s_{2}}$ and
$B_{s}$ respectively, meaning that they satisfy the same linear dependence relations. Then Equation \ref{eq} implies $$\displaystyle{\sum_{i=1}^{l_1} \mu_{i} C_{1}^{i} + \sum_{j=1}^{l_2} \delta_{j} C_{2}^{j} = \sum_{k=1}^{l_3} \lambda_{k} C_{3}^{k} }$$ 
We assume without loss of generality that $X_1$ admits a signature $\mu \neq \varnothing$. Therefore, there is a nontrivial linear combination of elements of $X_{1}$ equals
to  $\displaystyle{\sum_{i=1}^{l_1} \mu_{i} C_{1}^{i}}$.
The set $X_2$ admits the signature $\delta$, thus there is a linear combination of elements of $X_{2}$ equal to $\displaystyle{ \sum_{j=1}^{l_2} \delta_{j} C_{2}^{j}}$.
By summing the two linear combinations, we obtain a nontrivial linear combination of elements of 
$X_{1} \cup X_{2}$ equal to $\displaystyle{\sum_{i=1}^{l_1} \mu_{i} C_{1}^{i} + \sum_{j=1}^{l_2} \delta_{j} C_{2}^{j}}$  which is equal to $\displaystyle{\sum_{k=1}^{l_3} \lambda_{k} C_{3}^{k}}$ by the previous equality. 
\end{proof}

\begin{lemma}
\label{unbuild}
 Let $T$ be an enhanced tree, $s$ one of its nodes with children $s_{1}$, $s_{2}$
 and $N$ the label of $s$.
Let $X_{1}$ and $X_{2}$ be two sets of leaves chosen amongst the leaves of $T_{s_{1}}$ and $T_{s_{2}}$ respectively.
If $X = X_{1} \cup X_{2}$ admits $\lambda$ at $s$, then there are two signatures $\mu$ and $\delta$
such that $R(N, \lambda,\mu,\delta)$ holds, $X_1$ admits $\mu$ at $s_{1}$ and $X_{2}$ admits $\delta$ at $s_{2}$.
\end{lemma}

\begin{proof}
If $\lambda = \varnothing$, then the choice of $\mu=\delta= \varnothing$ proves the lemma.
Assume now that $X$ admits $\lambda \neq \varnothing$, hence there is a nontrivial linear combination of elements in $X$ equal 
to $v = \displaystyle{\sum_{k=1}^{l_3} \lambda_{k} C_{3}^{k}}$.
We can divide the linear combination of elements in $X$ into a sum of element in $X_1$ equal to $v_1$
and a sum of elements in $X_2$ equal to $v_2$ with $v = v_1 +v_2$.
At least one of those combinations is nontrivial, we assume it is the one equal to $v_1$.

Since $v_{1} = v - v_{2}$ and $v \in B_{s}$, we have $v_{1} \in \,\,<E_{s_{2}} \cup B_{s}>$. 
Moreover $B_{s} \subseteq E_{s}^{c} \subseteq E_{s_{1}}^{c}$ and $E_{s_{2}} \subseteq E_{s_{1}}^{c}$ then $v_{1} \in E_{s_{1}}^{c}$. 
Hence we have proven that $v_{1}$ is in $E_{s_{1}} \cap E_{s_{1}}^{c} = B_{s_{1}}$.
The vectors $C_{1}^{i}$ are a base of $B_{s_{1}}$, so that there is $\mu = (\mu_1,\dots,\mu_{l_1}) \neq \varnothing$ such that
$v_{1} = \displaystyle{ \sum_{i=1}^{l_1} \mu_{i}C_{1}^{i}}$. 
It means that  $X_{1}$ admits the signature $\mu$ at $s_1$.
Since $X_{2}$ plays a symmetric role, the same demonstration proves that it admits the signature $\delta =(\delta_{1}, \dots, \delta_{l_{2}})$  
at ${s_{2}}$ such that $v_{2} = \displaystyle{ \sum_{j=1}^{l_2} \delta_{j}C_{2}^{j}}$. 

Finally we have $\displaystyle{\sum_{i=1}^{l_1} \mu_{i} C_{1}^{i} + \sum_{j=1}^{l_2} \delta_{j} C_{2}^{j} = \sum_{k=1}^{l_3} \lambda_{k} C_{3}^{k} }$ and we can replace the columns of $C_{i}$ by those of $N_{i}$,
 which proves that $R(N, \lambda,\mu,\delta)$ holds.
\end{proof}

We then derive a global result on enhanced trees and signatures.

\begin{lemma}
  Let $A$ be a matrix representing a matroid and $T$ one of its enhanced tree.
Let $X$ be a set of leaves of $T$, then $X$ admits the signature $\lambda$ at the node $u$ 
if and only if there exists a signature $\lambda_{s}$ for each node $s$ of the tree $T$ such that:
  \begin{enumerate}
   \item for every node $s$ labeled by $N$ with children $s_{1}$ and $s_{2}$, $R(N, \lambda_s,\lambda_{s_1},\lambda_{s_2})$ holds.
   \item for every leaf $s$, $\lambda_{s} \neq \varnothing$ only if $s$ is in $X$ and $s$ is labeled by the matrix $(\alpha)$ 
with $\alpha \neq 0$.
     \item $\lambda_{u} = \lambda$.
  \end{enumerate}
\label{signature}
\end{lemma}

\begin{proof}
 The proof is by induction on the height of $s$ in $T$.
If $u$ is a leaf of $T$, the equivalence is true 
because of the second condition and  Example \ref{leaf}.

Assume now that $u$ is an internal node labeled by $N$ and with children $s_1$ and $s_2$.
The induction hypothesis and the conditions $1$ and $3$ enable us to use 
the Lemmas \ref{build} and \ref{unbuild} to prove both sides of the equivalence.
\end{proof}

The following theorem is the key to the next part, it shows that testing dependence of a set can be done by checking local 
constraints on signatures. 

\begin{theorem}[Characterization of dependence]
\label{finale}
  Let $A$ be a matrix representing a matroid, $T$ one of its enhanced tree and $l$ the bijection
between the leaves of $T$ and the columns of $A$. Let $X$ be a set of leaves of $T$,
  then $l(X)$ is dependent if and only if there exists a signature $\lambda_{s}$ for each node $s$ of the tree $T$ such that:
  \begin{enumerate}
   \item for every node $s$ labeled by $N$ with children $s_{1}$ and $s_{2}$, $R(N, \lambda_s,\lambda_{s_1},\lambda_{s_2})$ holds.
   \item for every leaf $s$, $\lambda_{s} \neq \varnothing$ only if $s$ is in $X$ and $s$ is labeled by the matrix $(\alpha)$ 
   with $\alpha \neq 0$.
   \item the signature at the root is $(0,\dots,0)$
  \end{enumerate}

\end{theorem}

\begin{proof}
If a set $X$ admits the signature $(0,\dots,0)$ at the root,
it means that there is a nontrivial linear combination of the elements of $l(X)$ equal to $0$.
It is therefore equivalent for $l(X)$ to be a dependent set and for $X$ to admit the signature $(0,\dots,0)$ at the root of $T$.
The proof of the theorem follows from this remark and Lemma \ref{signature} applied at the root.
\end{proof}

\subsection{Monadic second-order logic over terms and matroids}

\paragraph{Terms}
A functional signature is a pair $(F,A)$, where $F$ is a finite set of function symbols of positive arity
and $A$ is a finite set of constants. We denote by $T(F,A)$ the set of terms built over $(F,A)$.
Note that a term can be seen as a ranked tree of bounded degree: each internal node is labeled by an element of $F$, each leaf by an element of $A$.
In this article all the terms/trees are binary.

The terms of $T(F,A)$ are represented by a relational structure whose domain
is the set of nodes of the term. The structure has the binary relations $lchild(x,y)$ and
$rchild(x,y)$ which are true when $y$ is the left child, respectively the right child, of $x$.
We also have one unary relation for each symbol in $F$ and $A$, denoted by $label(s) = e$, which holds when 
$e$ is the label of the node $s$.

We recall the definition of monadic second-order logic, 
here given over terms, i.e. the atoms are made from the relations of the structure which represents a term.
The particularity of this logic is to use two types of variables.
A first order variable (in lower case) represents an element of the domain,
whereas a second-order variables (in upper case) represents a subset of elements
of the domain.

\begin{definition}
One builds atomic formulas from first and second-order variables and from
the relations $=$, $\in$, $rchild(x,y)$, $lchild(x,y)$ and  $label(s) = e$ for all $e$ of $A \cup F$.
The set of \textbf{M}onadic \textbf{S}econd \textbf{O}rder formulas, denoted by $MSO$,
is the closure of these atomic formulas by the usual quantifiers $\exists$, $\forall$ and
 the logical connectives $\wedge$, $\vee$ and $\neg$.
\end{definition}

The equality is the equality over the elements of the domain, but we extend it to 
sets, since it is definable by a simple formula. The relation $x \in X$ means that the element 
denoted by $x$ is a member of the set denoted by $X$. We also use freely 
$\neq$ and $\subseteq$ over elements and sets since they are easily definable in $MSO$.
We can express by a $FO$ formula the fact to be the root or a leaf:
$$root(s)\equiv \forall x \neg (lchild(x,s) \vee rchild(x,s))$$
$$leaf(s)\equiv \forall x \neg (lchild(s,x) \vee rchild(s,x))$$

One can decide if an $MSO$ formula holds over a term by building an appropriate tree automaton and running it on the term.
This yields the following classical theorem.

\begin{theorem}[Thatcher and Wright \cite{thatcher1968gfa}]
\label{folklore}
The model-checking of $MSO$ formulas over terms is solvable by a fixed parameter linear algorithm, 
the parameter is the sum of the size of the formula and the
size of the functional signature on which the terms are defined. 
\end{theorem}

\paragraph{Matroids}

A matroid is represented by a structure whose domain is the ground set of elements of the matroid.
There is a second-order relation in the structure, denoted by $indep(X)$, which holds when the set $X$ is independent in the matroid.
Notice that, since $indep$ is a set predicate, we are not in the usual framework of first-order relational structures.
We define the monadic second-order logic structure representing matroids, exactly as for terms,
but with the relations $=$, $\in$ and $indep$. This logic is denoted by $MSO_{M}$.
We now give some properties definable in this logic.
For more details and examples, one may read \cite{hlineny2003matroid}.

The circuits are definable in $MSO_{M}$, $X$ is a circuit if and only if it satisfies:
$$circuit(X) \equiv \neg indep(X) \wedge \forall Y\left( Y \nsubseteq X \vee X = Y \vee indep(Y)\right) $$

We can also express that a matroid is connected, meaning that every pair of elements is in a circuit,
a notion similar to $2$-connectivity in graphs:
$$ \forall x,y\, \exists X\, x \in X \wedge y \in X \wedge circuit(X)$$

The axioms defining a matroid in term of circuits, given at the end of Sec. \ref{sec:matroid}, are also expressible in $MSO_M$.

One defines the notion of minor of a matroid
by using the restriction presented in Sec.\ref{sec:matroid} and an operation of contraction.
For any matroid $N$, one can write a formula $\psi_N$ of $MSO_{M}$ which is true on a matroid $M$ if and only if 
$N$ is a minor of the matroid $M$ (see \cite{hlineny2003matroid}). Therefore all classes of matroids defined by excluded minors, 
such as the matroids representable over $\mathbb{F}_2$ \cite{tutte1958homotopy}, $\mathbb{F}_3$ \cite{seymour1979matroid} or $\mathbb{F}_4$ \cite{geelen2000excluded},
 are also definable by an $MSO_{M}$ formula.\footnote{These examples are partial results towards Rota's conjecture, that is to prove that the matroids representable over a finite field
can be characterized by excluded minors for any finite field. }

One can express some properties about a graph by a formula over its cycle matroid.
For instance, one can check that a graph is Hamiltonian if and only if it has a cycle 
containing a spanning tree. This can be stated by the next formula, where $basis(X)$ 
is a formula which holds if and only if $X$ is a basis:
$$ \exists C \, circuit(C) \wedge \exists x \,basis(C \setminus \{x\})$$

In fact, it has been proven in \cite{hlinveny2006bwp} that any sentence about a loopless $3$-connected graph $G$
in $MS_2$ can be expressed as a sentence about its cycle matroid in $MSO_{M}$.
This property can be generalized to any graph, by considering the cycle matroid of $G\uplus K_3$
which is a disjoint union of $G$ and $K_3$ with all edges between the two graphs.

\subsection{From Matroids to Trees}
\label{formule}

The aim of this subsection is to translate $MSO_{M}$ formulas over a matroid into
$MSO$ formulas over its enhanced tree.
The main difficulty is to express the predicate $indep$ in $MSO$. To achieve that, we use  Theorem \ref{finale}
which reduces this property to an easily checkable condition on a signature at each node of the enhanced tree. 
This can be seen as finding an accepting run of a non deterministic automaton whose states are signatures.

The formula is defined for enhanced trees of width less than $t$ over a field $\mathbb{F}$ of size $k$. 
We have to encode in $MSO$ a signature $\lambda$ of size at most $t$ at each node of an enhanced tree.
These signatures are represented by the set $\vec{X}$ of set variables $X_{\lambda}$ indexed by all signatures $\lambda$ of size at most $t$.
The number of such variables is bounded by $(k+1)^{t}$, a constant because both the field and the branch-width are fixed.

 The relation $X_{\lambda}(s)$ holds if and only if $\lambda$ is the signature at $s$. 
The following formula states that there is one and only one value for the signature at each $s$.
 $$\displaystyle{\Omega(\vec{X}) \equiv \forall s \, \bigvee_{\lambda}\left( X_{\lambda}(s) \bigwedge_{\lambda' \neq \lambda}\neg X_{\lambda'}(s)\right)} $$

\begin{remark}
We could have defined the signature of a set as the union of all the signatures it admits, 
as we do in Sec. \ref{sec:terms}. The signature would then be unique, and our construction
would correspond to a deterministic automaton. But in this case, we would deal with $2^{k^{t}}$ possible
signatures, a number still bounded if $k$ and $t$ are fixed, but which is much larger and further decreases the practical interest 
of the algorithm we provide.

If we want to be more efficient and use less variables, we may encode in binary the value 
of each element $\lambda_{i} \in \mathbb{F}$ of a signature $\lambda$. We only need $\log(k)t$ variables 
to do so and it also spares us the formula $\Omega$ but it would obfuscate the presentation.
\end{remark}

The formula $dep(Y)$ that represents the negation of the relation $indep$ is now built in three steps corresponding
to the three conditions of Theorem \ref{finale}.

\begin{enumerate}
 \item 

The formula $\Psi_{1}$ ensures that the relation $R$ holds at every internal node.
It is a conjunction on all possible characteristic matrices $N$ and all signatures $\lambda$.

$$\Psi_{1}(\vec{X}) \equiv \forall s \, \neg leaf(s) \Rightarrow [ \exists s_{1}\, s_{2} \, lchild(s,s_{1}) \wedge rchild(s,s_{2})$$
$$ \displaystyle{\bigwedge_{\lambda_{1},\lambda_{2},\lambda,N}   \left( \left( label(s) = N \wedge  X_{\lambda_{1}}(s_{1}) \wedge  X_{\lambda_{2}}(s_{2}) \wedge  X_{\lambda}(s)\right) \Rightarrow  R(N,\lambda,\lambda_{1},\lambda_{2}) \right)]}$$
 
\item 
We define the formula  $\Psi_{2}(Y,\vec{X})$ which means that a leaf with a signature different from $\varnothing$ is 
in $Y$ and has a label different from the matrix $(0)$.
$$\Psi_{2}(Y,  \vec{X}) \equiv \forall s \left( leaf(s) \wedge \neg X_{\varnothing}(s)\right) \Rightarrow ( Y(s) \wedge label(s) \neq (0))$$
\item

$\Psi_{3}(\vec{X})$ states that the signature at the root is $(0,\dots,0)$.
$$ \Psi_{3}(\vec{X}) \equiv \exists s \, root(s) \wedge  X_{(0,\dots,0)}(s) $$
\end{enumerate}

Thanks to Theorem \ref{finale} we know that the following formula is true 
on an enhanced tree $T$ of a matroid $M$ if and only if $Y$ is a set of leaves of $T$ in bijection with a dependent set of $M$.
$$dep(Y) \equiv  \exists \vec{X}  \,\Omega(\vec{X}) \wedge \Psi_{1}( \vec{X} ) \wedge \Psi_{2}(Y, \vec{X})\wedge \Psi_{3}(\vec{X}) $$

The size of the formula $dep(Y)$ is  up to a constant factor the size of $\Psi_1$ which is a conjunction of less than $k^{9t^{2}+3t}$ terms of constant size
plus the size of $\Omega$ which is disjunction of $k^t$ terms of size $k^t$. Therefore, when $k$ and $t$ are fixed, $dep$ is of fixed size.

We now define by induction  a formula $F(\phi(\vec{x}))$ of $MSO$ from the formula $\phi(\vec{x})\in MSO_{M}$, by relativization to the leaves.
\medskip

\begin{itemize}
 \item if $\phi(\vec{x})$ is the relation $x=y$ or $x \in X$,  $F(\phi(\vec{x}))$ is the same relation 
 \item if $\phi(\vec{x})$ is the relation $indep(X)$, $F(\phi(\vec{x}))$ is the negation of the formula $dep(X)$ we have just defined
  \item if $\phi(\vec{x})$ is the formula $\psi(\vec{x}) \wedge \chi(\vec{x})$, $F(\phi(\vec{x}))$ is the formula $F(\psi(\vec{x})) \wedge F(\chi(\vec{x}))$ 
 \item if $\phi(\vec{x})$ is the formula $\exists y \psi(y)$, $F(\phi(\vec{x}))$ is the formula $\exists y (leaf(y) \wedge F(\psi(y)))$
 \item if $\phi(\vec{x})$ is the formula $\exists Y \psi(Y)$, $F(\phi(\vec{x}))$ is the formula $\exists Y [\forall y (y \in Y \Rightarrow leaf(y)) \wedge F(\psi(Y))]$
\end{itemize}
\medskip

Moreover, for every free first-order variable $y$ and every free second-order variable $Y$,
we take the conjunction of the relativized formula above with:
\begin{itemize}
 \item  $leaf(y)$ 
\item $\forall y (y \in Y \Rightarrow leaf(y))$
\end{itemize}
We can now state the main theorem:

\begin{theorem}
\label{main}
 Let $M$ be a $\mathbb{F}$-matroid of branch-width less than $t$, $T$ one of its enhanced tree 
and $l$ the bijection between the leaves of $T$ and the elements of $M$.
Let $\phi(\vec{x})$ be a $MSO_{M}$ formula with free variables $\vec{x}$,
we have $$(M,\vec{a}) \models \phi(\vec{x})  \Leftrightarrow (T,l(\vec{a})) \models F(\phi(\vec{x}))$$
\end{theorem}

\begin{proof}
 The demonstration is done by induction, every case is trivial except the translation of the
 predicate $indep$ whose correctness is given by Theorem \ref{finale}.
\end{proof}

Suppose we have a formula $\phi$ of $MSO_{M}$ and a representable matroid $M$ of branch-width $t$.
We know that we can find a branch-width decomposition of width equal to $t$ in cubic time \cite{hlineny2008finding}.
Furthermore, we can build from it an enhanced tree of $M$ in cubic time. By Theorem \ref{main}, we know that we need only to decide 
the formula $F(\phi)$ on the enhanced tree to decide $\phi$ on $M$, which is done in linear time by Theorem \ref{folklore}.
We have, as a corollary, the main result of \cite{hlinveny2006bwp}.

\begin{corollary}[Hlin\v{e}n{\'y} \cite{hlinveny2006bwp}]
 The model-checking problem of $MSO_{M}$ formulas is decidable in time $f(t,k,l) \times n^{3}$ over the set of  $\mathbb{F}$-matroids given by a matrix,
where $n$ is the number of elements in the matroid, $t$ is its branch-width, $k$ is the size of  $\mathbb{F}$, $l$ is the size of the formula and $f$ is a computable function.
\end{corollary}

Since we can decide dependence in a represented matroid of bounded branch-width in linear time by only using one of its enhanced tree,
the enhanced trees are a way to describe completely a matroid and then to represent it. Moreover, 
this representation is compact, since the size of an enhanced tree is $O(t^{2}\times n)$, where $n$
is the size of the ground set of the matroid, while the matrix which usually defines it, is of size  $O(n^{2})$.

\section{Extensions and Applications}

In this section, we present generalizations of the result of the previous section,
by an extension of the model or of the language.
As an application, we show that Theorem \ref{main} can be used to solve enumeration problems in a more efficient way.

\subsection{Logical extension}

\paragraph{Colored matroids}

 We can work with colored matroids, meaning that we add a fixed number of unary predicates to the language
 which are interpreted by subsets of the ground set.
Theorem \ref{main} still holds for colored matroids except that we now have colored trees, on which the decision problem for
$MSO$ is still in linear time.

 Let $\textsc{A-Circuit}$ be the problem to decide, given a matroid $M$ and a subset $A$ of its elements, if there is a circuit in which $A$ is included.
This problem is interesting, since when $|A|=1$ and the matroid is representable over a finite field,
a circuit extending $A$ is a minimal solution (for inclusion of the support) of a linear system.
If the field is $\mathbb{F}_{2}$, a circuit extending $A$ is a minimal solution (for the pointwise order) of an affine formula.
 It is an affine variation of the circumscription problem for propositional formulas studied in artificial intelligence \cite{mccarthy1980circumscription}. 
\begin{enumerate}
 \item If $|A| = 1$ or $2$, the problem is decidable in polynomial time. For the particular case of a vector matroid see \cite{durand2003inference},
in general one uses a matroid separation algorithm.
\item If $|A| = 3$, the question is open.
\item If $|A|=k$ is fixed and the matroid is a cycle matroid then it is decidable in polynomial time by reduction to the problem of finding $k$ disjoint paths in a graph \cite{robertson1995graph}.
\item If $|A|$ is unbounded, even if the matroid is only a cycle matroid, the question is $\NP$-complete by reduction from the Hamiltonian Path problem.
\end{enumerate}
This problem is easily expressible in $MSO_{M}$ over a matroid equipped with a unary second-order predicate $A$, by the formula $A-Circuit(X) \equiv A \subseteq X \wedge Circuit(X)$.
Thus $\textsc {A-Circuit}$ is decidable in polynomial time over representable matroids of branch-width $t$,
while it is a $\NP$-complete problem in general.

\paragraph{Counting $MSO$}

The second generalization is to add to the language a fixed number of second-order predicates $Mod_{p,q}(X)$ which mean that $X$ is of size $p$ modulo $q$.
We obtain the logic called $CMSO_{M}$ for counting monadic second-order. In this logic, we can express 
the fact that a set is a circuit of even cardinality, which is not possible in $MSO_{M}$.
Theorem \ref{main} also holds for $CMSO_{M}$ except that the translated formula is now in $CMSO$.
This is interesting since the model-checking of $CMSO$ is solvable in linear time over trees \cite{courcelle1992monadic}.

We could also adapt Theorem \ref{main} to $MSO_{M}$ problems
with optimization constraints, that is finding the minimal or maximal size of a set which satisfies a formula.
This kind of problem has been introduced in \cite{arnborg1991easy} for graphs under the name of $EMSO$.
These problems are solvable in linear time for graphs of bounded tree-width.
For instance, using the formula $A-Circuit(X)$, we can find the size of the minimum circuit which extends a set $A$.
When the matroid is binary and $|A|=1$, it is equivalent to the problem of finding the minimum
weight of a solution of an affine formula, which is $\NP$-complete \cite{berlekamp1978inherent}. 

\subsection{Enumeration}

Let us first define enumeration problems and the associated complexity measures.
Let $A$ be binary predicate over $\Sigma^{*}$ where $\Sigma$ is finite alphabet. One says that $A$ 
is \emph{polynomially balanced} when there is a polynomial
$Q$ such that if $A(x,y)$ holds then $|y| < Q(|x|)$. We write $A(x)$ for the finite set $\{y\mid A(x,y)\}$.
The enumeration problem associated to $A$, denoted by $\enum{A}$,
consists in computing the function which associates $A(x)$ to $x$. 
 
An enumeration algorithm does not output the whole set $A(x)$ and stops:
it outputs the elements of $A(x)$ one after the other. The measure of complexity, out of the total time to output 
all elements, is the time between the output of one solution and the next, which is called the \emph{delay}.
We say that a problem $\enum{A}$ is solvable by an algorithm in \emph{incremental delay}
if, for all inputs $x$, its delay between the $i^{\text{th}}$ and the $i+1^{\text{th}}$ solutions is polynomial in $|x|$ and $i$.
If its delay is polynomial in $|x|$ only, we say that the algorithm is in \emph{polynomial delay}.

We now present a theorem, which gives algorithms in polynomial delay to solve a lot of problems
on matroids. We use it specifically to solve the problem $\enum{\textsc{A-Circuit}}$ over matroids of bounded branch-width
representable on finite fields.

\begin{theorem}[Courcelle \cite{courcelle2009linear}]
Let $\phi(X_{1},\dots ,X_{m})$ be an $MSO$ formula,
there exists an enumeration algorithm which given a term $T$ of size $n$ and of depth $d$ enumerate the $m$-tuples $ B_{1},\dots, B_{m}$
such that $ T \models \phi(B_{1},\dots, B_{m})$ with a linear delay and a preprocessing time $O (n \times d) $.
\label{courcelle}
\end{theorem}

The next corollary is a direct consequence of Theorems  \ref{main} and \ref{courcelle}.

\begin{corollary}
\label{enumeration}
Let  $\phi(X_{1},\dots ,X_{m})$ be an $MSO_{M}$ formula, let $t$ be an integer and let $\mathbb{F}$ be a field.
There is an algorithm, which given a $\mathbb{F}$-matroid $M$ of branch-width less than $t$,
enumerates the $m$-tuples $ B_{1},\dots, B_{m}$ such that $ M \models \phi(B_{1},\dots, B_{m})$ with a linear delay after a cubic preprocessing time.
\end{corollary}

\begin{proof}
 Let $\phi(X_{1},\dots ,X_{m})$ be an $MSO_{M}$ formula, we compute in constant time $F(\phi(X_{1}, \dots, X_{m}))$, the formula for matroids of branch-width at most $t$ given by Theorem \ref{main}.
Then, given a matroid of branch-width $t$, we compute its enhanced tree in cubic time. We run the enumeration algorithm
given by Theorem \ref{courcelle} on this enhanced tree and the formula $F(\phi(X_{1} ,\dots, X_{m}))$.
Each time we find a $m$-tuple satisfying the formula, we output its image by the bijection
between the leaves of the enhanced tree and the elements of the matroid. This algorithm gives
the solutions of $\phi(X_{1},\dots ,X_{m})$ with a linear delay and a cubic preprocessing time.
\end{proof}

Our example of the previous subsection, the problem $\textsc{A-Circuit}$, yields the interesting enumeration problem $\enum{\textsc{A-Circuit}}$.
This problem admits an algorithm in incremental delay \cite{1113216} when $|A|=1$ and the matroid has an independence predicate decidable in polynomial time.
We would like to have an algorithm for this problem with polynomial delay rather than incremental.
The only known result in this vein is for $|A|$ of fixed size and cycle matroids \cite{read1975bounds}.

Corollary \ref{enumeration} can be adapted to $MSO_{M}$ over colored matroids and thus applied to the formula $A-Circuit(X)$.
We obtain an algorithm in linear delay, which solves $\enum{\textsc{A-Circuit}}$ on matroids representable over a finite field and of branch-width $t$.
In addition to its good delay, this algorithm is the first which solves the problem for an unbounded $A$. Moreover, the time it takes to output all solutions is linear
in the number of solutions, while the incremental algorithm of \cite{1113216} needs a time cubic in this number.
Another polynomial delay algorithm for $\enum{\textsc{A-Circuit}}$ on a class of very ``dense'' representable matroids is also presented in the second chapter of \cite{phd_strozecki}.

\section{Matroid Operations}
\label{sec:terms}

In this section, we give two different ways to build matroids by means of some well chosen operations.
We then prove that the model-checking of $MSO_M$ is decidable in linear time on these classes of matroids. Definitions and notations are inspired from
\cite{hlinveny2006bwp} and are sometimes slightly modified to deal with different matroid grammars.

\subsection{Pushout of boundaried matrices}

\begin{definition}[Boundaried matroid]
A pair $(M,\gamma)$ is called a $t$ boundaried matroid if $M$ is a matroid and 
$\gamma$ is an injective function from $ \left[1,t \right]$ to $M$ whose image is an independent set.
The elements of the image of $\gamma$ are called boundary elements and the others are called internal elements.
\end{definition}

The restriction of $M$ to its ground set minus the elements of the boundary is called the internal matroid of  $(M,\gamma)$.
We need an operation $\oplus$, which associates a matroid $N_{1} \oplus N_{2}$ to two $t$ boundaried matroids $ N_{1}$ and $N_{2}$.
By means of this operation, we try to properly define a set of terms similar to those introduced in \cite{hlinveny2006bwp}. 
Hereafter, we explain how these terms are related to enhanced trees.
The same technique will be used with a different operation in the next section. 

A $t$ boundaried matrix is a matrix $A$ and an injective function $\gamma$ from $ \left[1,t \right]$
to $A$ whose image is an independent set of columns.
Boundaried matrices represent boundaried matroids in the obvious way.
In fact, we define the operation $\oplus$ on boundaried matrices and not on the boundaried matroids they represent.

We want to define $\oplus$ as the pushout (or amalgam) of two boundaried matrices.
 It would then generalize the construction of decomposition trees for graphs of bounded branch-width, also obtained by a pushout in the category of graphs. 
By pushout, we mean an operation  such that  $A_1$ and $A_2$ can be injected in 
$A_{1} \oplus A_{2}$ by the morphisms $i_1$ and $i_2$ respectively and such that 
$i_1(\gamma_1(j)) = i_2(\gamma_2(j))$ for all $j$.
We present a way to define such a pushout between two vector spaces, which
yields an operation on boundaried matrices.

Let $(A_{1},\gamma_{1})$ and $(A_{2},\gamma_{2})$ be two $t$ boundaried matrices over the same field $\mathbb{F}$.
We see $A_{i}$ as a set of vectors in the vector space $E_{i}$.
Let $E_{1} \times E_{2}$ be the direct product of the two vector spaces and  let $B$ be its subspace generated 
by the elements $(\gamma_{1}(j), -  \gamma_{2}(j)) $ for all $j$.

\begin{definition}
Let $E$ be the quotient space of $\left( E_{1} \times E_{2}\right)$ by 
$B$. We write  $(A_{1},\gamma_{1}) \oplus(A_{2},\gamma_{2})$ the set of vectors in $E$ 
of the form $(a_1,0)$ with $a_1 \in A_{1}\setminus \gamma_1([1,t])$ and $(0,a_2)$ with $a_2 \in A_{2}\setminus\gamma_2([1,t])$.
\end{definition}

Remark that $(A_{1},\gamma_{1}) \oplus(A_{2},\gamma_{2})$ defines a (non boundaried) $\mathbb{F}$-matroid.
To have a more specific idea of the action of $\oplus$ and give examples, 
we must explain how to unambiguously represent $(A_{1},\gamma_{1}) \oplus(A_{2},\gamma_{2})$ by a matrix.
Since, once a base is chosen, a set of vectors and a matrix are the same objects, 
we only have to give an algorithm to build a base of $E$.
We build a base $B$ of $E$ from $C$ and $D$, the canonical bases of $E_1$ and $E_2$. 
Let $i_1$ (respectively $i_2$) be the injection from $E_1$ to $E$ (resp. from $E_2$ to $E$). Let $B_{0} =  i_1(C) $ and $B_{j+1} = B_{j} \cup \left\lbrace i_2(D_{j+1})\right\rbrace$
 if this set is independent, otherwise $B_{j+1} = B_{j}$. Let $n$ be the size of $C_{2}$, then $B$ is $B_{n}$, which is by construction 
 a base of $E$.

\begin{example}
\label{counterexample}
\begin{small}

$$
\begin{array}{c c c c c}
  \left( \begin{array}{c c  |  c}
1 & 0\, & \,1 \\
0 & 1\, & \,1 \\
 \end{array} \right) &
\oplus &
 \left( \begin{array}{c c | c c}
1 & 0\, & \,1 & 0\\
0 & 1\, & \,1 & 0\\
0 & 1\, & \,1 & 1
 \end{array} \right)  &
= &
 \left( \begin{array}{c c c}
1 & 1 & 0\\
1 & 1 & 1\\
0 & 0 & -1
 \end{array} \right) 
 \end{array}
$$
$$
\begin{array}{c c c c c}
  \left( \begin{array}{c c| c}
1 & 0\, & \,1 \\
0 & 1\, & \,1 \\
 \end{array} \right) &
\oplus &
 \left( \begin{array}{c c | c c}
1 & 0\, & \,2 & 0\\
0 & 1\, & \,1 & 0\\
0 & 1\, & \,1 & 1
 \end{array} \right)  &
= &
 \left( \begin{array}{c c c}
1 & 2 & 0\\
1 & 1 & 1\\
0 & 0 & -1
 \end{array} \right) 
 \end{array}
$$
\end{small}
The matrices of the example can be seen as defined over $\mathbb{F}_{3}$ or any larger field.
The boundary elements are the two first columns of the matrices, separated from the others by the symbol $|$ for clarity.
 The image of the canonical base of $E_{1}$ in $E$ is $\{e_{1},e_{2}\}$ and the image of $E_2$ is $\{e_{3},e_{4},e_{5}\}$.
By identification of the first and second columns, we have $e_{1}=e_{3} $ and $e_{2} =e_{4} + e_{5} $.
The basis built by the algorithm is thus $\left\lbrace e_{1}, e_{2}, e_{4}\right\rbrace$.\\
The column $(1,1,1)^{t}$ of the second matrix in the left hand side of the first equation is represented in the right hand side by  
$(1,1,0)^{t}$. Indeed, once injected in $E$, this vector is equal to $e_{3}+e_{4}+e_{5}$ which is equal to $e_{1}+e_{2}$,
 the sum of the two first vectors of the base we have built.

Notice that the columns $1$ and $2$ of the result in the first equation form a dependent set but not in the result of the second,
thus the two matrices obtained represent distinct matroids.
Yet the matrices we combine by $\oplus$, although different, represent the same matroid in both equations. 
\end{example}

Example  \ref{counterexample} shows that $\oplus$ cannot be seen as an operation on matroids because the result depends on the way
the matroids are represented. We could also make this kind of construction by representing matroids by projective spaces, as it is done in \cite{hlinveny2006bwp}.
Unfortunately, we would define essentially the same operation, which would still be defined over the projective spaces and not the matroids.
Nevertheless, if we restrict $\oplus$ to matrices over $\mathbb{F}_{2}$, it properly defines an operation on the matroids they represent.

\begin{proposition}
 Let $(M_{1},\gamma_{1})$ and $(M_{2},\gamma_{2})$ be two boundaried $\mathbb{F}_{2}$-matroids. For all 
  matrices $A_{1}$ and $A_{2}$ which represents these matroids, the matroid represented by $A_{1} \oplus A_{2}$ 
 is the same.
\end{proposition}

\begin{proof}
We are going to show that the fact to be a circuit of  $A_{1} \oplus A_{2}$ depends
only on $M_1$ and $M_2$. Since a matroid is entirely determined by its set of circuits, it will prove the proposition.

 A circuit of $A_{1} \oplus A_{2}$ is the union of internal elements of $A_{1}$ and $A_{2}$ denoted by 
 $X$ and $Y$ such that $\displaystyle{\sum_{x \in X} (x,0) + \sum_{y \in Y}(0,y)} \in \left\langle \left\lbrace (\gamma_{1}(j), -  \gamma_{2}(j))\right\rbrace_{j\leq t}  \right\rangle$ and $X \cup Y$ is minimal for this property. Equivalently, there is a set $S \subseteq [1,t]$ such that 
 the two following relation hold: 
 \begin{itemize}
  \item $\displaystyle{\sum_{x \in X} x + \sum_{i \in S} \gamma_1(i) = 0}$
  \item  $\displaystyle{\sum_{y \in Y} y + \sum_{i \in S} \gamma_2(i) = 0}$
 \end{itemize}
This is true because, the matrices $A_1$ and $A_2$ are defined over $\mathbb{F}_{2}$, therefore all coefficients
different from zero have to be one.
 It is equivalent to:  $X \cup \gamma_1(S)$ is a circuit of  $A_1$, thus of $M_1$ and $Y \cup \gamma_2(S)$ is a circuit of $A_2$ thus of $M_2$
\end{proof}

Behind this proof is hidden the notion of the signature of a set in a boundaried matroid that we are going to use afterwards.
We now want to build matroids from successive applications of the operation $\oplus$. 

\begin{definition}
Let $A$ be a matrix and let $\gamma^{A}_{i}$ for $i=1,2,3$ be
three injective functions from $\left[1,t_{i}\right]$ to the columns of $A$.
If the sets $\gamma^{A}_{i}(\left[1,t_{i}\right])$ are independent
and form a partition of the columns of $A$,
then $(A,\{\gamma^{A}_{i}\}_{i=1,2,3})$ is called a \emph{$3$-partitioned matrix}.

\noindent
Let $M$ be a matroid and let $\gamma^{M}_{i}$ for $i=1,2,3$ be
three injective functions from $\left[1,t_{i}\right]$ to the ground set of $M$.
If the sets $\gamma^{M}_{i}(\left[1,t_{i}\right])$ are independent
and form a partition of the columns of $M$,
then $(M,\{\gamma^{M}_{i}\}_{i=1,2,3})$ is called a \emph{$3$-partitioned matroid}.
\end{definition}

The characteristic matrices used to build the enhanced trees may be seen as $3$-partitioned matrices.
From $\oplus$ and $A$ a $3$-partitioned matrix we define an operator $\odot_{A}$ which associates a boundaried matrix to two boundaried matrices.
It is defined by two successive uses of $\oplus$ on the boundaries $\gamma^{A}_{1}$ and $\gamma^{A}_{2}$.

\begin{definition}
\label{def:odot}
Let $\overline{A_{1}}=(A_{1},\gamma_{1})$ and $\overline{A_{2}}=(A_{2},\gamma_{2})$  be respectively a $t_{1}$ and a $t_{2}$ boundaried matrix
and let $A$ be a $3$-partitioned matrix.
 We call $ \overline{A_{1}} \odot_{A} \overline{A_{2}}$ the $t_{3}$ boundaried matrix defined by $ (\overline{A_{1}} \oplus ( A,  \gamma^{A}_{1}), \gamma^{A}_{2}) \oplus \overline{A_{2}}$ with boundary $\gamma^{A}_{3}$.
\end{definition}

\begin{figure}[ht]
 \centering
 \includegraphics[scale=0.4]{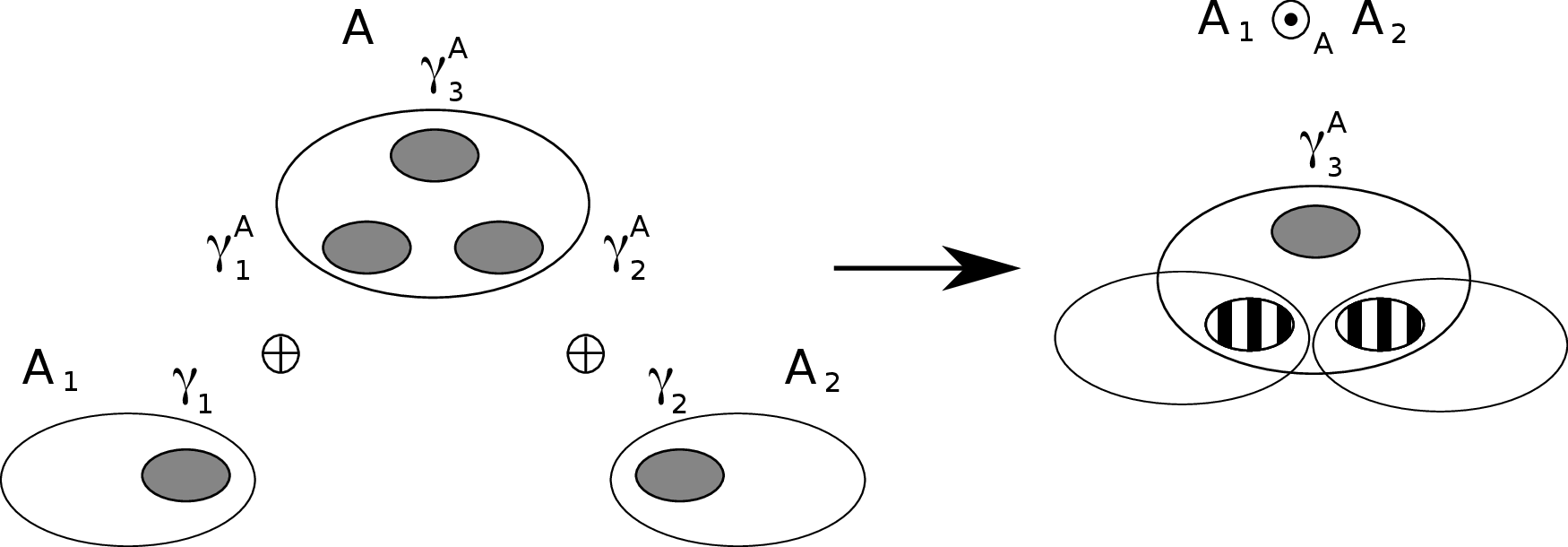}
 \caption{Representation of the operation $\odot$, boundaries represented in grey and removed parts hatched}
\end{figure}  

The operation $\oplus$ is ``associative'' meaning that $\overline{A_{1}}\odot_{A}\overline{A_{2}}$
 can also be defined by  $ \overline{A_{1}} \oplus (( A,\gamma^{A}_{2}) \oplus \overline{A_{2}}, \gamma^{A}_{1})$ with boundary $\gamma^{A}_{3}$.

Let $\Upsilon$ be the set containing the two following $1$-boundaried matrices:
\begin{itemize}
 \item $\varUpsilon_{0}$ is the matrix $\left(\begin{array}{c | c}
 1\, & \,0 \\
0\, & \,1
\end{array}\right)
 $ .
\item $\varUpsilon_{1}$ is the matrix $\left( \begin{array}{c | c}
 1\, & \,1 
\end{array}\right) 
 $ .
\end{itemize}
\begin{definition}
 Let $\mathcal{M}^{\mathbb{F}}_{t}$ be the set of terms which are inductively defined by:
 \begin{itemize}
  \item An element of $\varUpsilon$ is a term of $\mathcal{M}^{\mathbb{F}}_{t}$.
  \item Let $T_{1}$ and $T_{2}$ be two terms of value $A_{1}$ and $A_{2}$ which are a $t_{1}$ and a $t_{2}$ boundaried matrix.
  Let $A$ be a $3$-partitioned matrix, its three parts being of cardinality $t_{1}$, $t_{2}$ and $t_{3}$, \emph{all less than or equal to $t$}.
  Then $A_{1} \odot_{A} A_{2}$ is a term of  $\mathcal{M}^{\mathbb{F}}_{t}$ whose value is a $t_{3}$ boundaried matrix.
 \end{itemize}
\end{definition}

The value of a term of $\mathcal{M}^{\mathbb{F}}_{t}$ is a matrix with a boundary. We will not distinguish
 a term from its value and the matroid it represents when we remove the boundary.
 To study the matroids represented by these terms, we now need to define the signature of a set $X$ in 
 exactly the same way as for enhanced trees.

\begin{definition}
 Let $T$ be a term of $\mathcal{M}_{t}^{\mathbb{F}}$ and let $(A,\gamma^{A})$ the boundaried matrix defined by $T$.
 We write $l$ for the size of the boundary of $A$.  Let $X$ be a subset of internal elements of $A$. 
 We say that $X$ admits the signature $\lambda = (\lambda_{1}, \dots, \lambda_{l})$ in $T$
 if there is a nontrivial linear combination of its elements equal to $\displaystyle{\sum_{i\leq l} \lambda_{i}\gamma^{A}(i)}$.
 The set $X$ always admits the signature $\varnothing$.
 \end{definition}

We now show that the signatures in a term of $\mathcal{M}_{t}^{\mathbb{F}}$ satisfy the relation $R$
given in Definition \ref{def:relation}. To this aim, we prove two lemmas similar to Lemmas \ref{build} and \ref{unbuild} in
which we use the following notations:
\begin{itemize}
 \item Let $T$ be a term of $\mathcal{M}_{t}^{\mathbb{F}}$ equal to $T_1 \odot_A T_2$, where
$A$ is a $3$-partitioned matrices.
\item The terms $T_1$, $T_2$ and $T$ represent the $t_1$, $t_2$ and $t_3$ boundaried matrices $(N_1,\gamma_1)$, $(N_2,\gamma_2)$ and $(N,\gamma_3)$.
\item Let $E_1$, $E_2$ and $E_3$ be the vector spaces generated by the columns of $N_1$, $N_2$ and $A$.
\item Let $V$ be the vector space $E_1\times E_2 \times E_3$.
\item Let $B$ be $\left\langle \left\lbrace (\gamma_{1}(j),0,0) -  (0,0,\gamma_{1}^A(j))\right\rbrace  \right\rangle$ and $C$ be 
$\left\langle \left\lbrace (0,\gamma_{2}(j),0) -  (0,0,\gamma_{2}^A(j))\right\rbrace  \right\rangle$.
\item Let $E$ be the quotient of $V$ by $B$ and then by $C$, it is the vector space which is used to define $(N,\gamma_3)$.
\item Let $\phi_1$, $\phi_2$ and $\phi_3$ be the injection of $E_1$, $E_2$ and $E_3$ in $E$.
\end{itemize}

\begin{lemma}
\label{buildbis}
Let $X_{1}$ and $X_{2}$ be two sets of internal elements of $N_1$ and $N_2$.
If $X_1$ admits $\mu$ in $T_1$, $X_{2}$ admits $\delta$ in $T_2$ and 
 $R(A, \lambda,\mu,\delta)$ holds then $X = X_{1} \cup X_{2}$ admits $\lambda$ in $T$. 
\end{lemma}

\begin{proof}
By definition of the signature, we know that a nontrivial combination 
of internal elements of $X_1$ (respectively of $X_2$) is equal to $\displaystyle{\sum_{1 \leq i \leq t_1} \mu_i \gamma_1(i)}$ 
(respectively to  $\displaystyle{\sum_{1 \leq j \leq t_2} \delta_j \gamma_2(j)}$).
 Therefore, there is a combination of elements of $X_1 \cup X_2$ 
seen as elements of $E$ which we write $v$ and which satisfies:
$$
v = \displaystyle{\phi_1\left(\sum_{1 \leq i \leq t_1} \mu_i \gamma_1(i)\right)} + \displaystyle{\phi_2\left(\sum_{1 \leq j \leq t_2} \delta_j \gamma_2(j)\right)}
$$
Since $\phi_1(\gamma_1(i)) = \phi_3(\gamma_1^A(i))$ and $\phi_2(\gamma_2(i)) = \phi_3(\gamma_2^A(i))$ for all $i$,
 $$v  = \displaystyle{\phi_3\left(\sum_{1 \leq i \leq t_1} \mu_i \gamma_1^A(i)\right) + \phi_3\left(\sum_{1 \leq j \leq t_2} \delta_j \gamma_2^A(j)\right)}$$
Moreover, $\phi_3$ is a linear function, therefore we have:
$$v = \displaystyle{\phi_3\left(\sum_{1 \leq i \leq t_1} \mu_i \gamma_1^A(i) + \sum_{1 \leq j \leq t_2} \delta_j \gamma_2^A(j)\right)}$$
Because $R(A,\lambda,\mu,\delta)$ holds, we have the equality 
$$\displaystyle{\sum_{1 \leq k \leq t_3} \lambda_k \gamma_3^A(k) = \sum_{1 \leq i \leq t_1} \mu_i \gamma_1^A(i) + \sum_{1 \leq j \leq t_2} \delta_j \gamma_2^A(j)} $$
This equality yields 
 $$v = \displaystyle{\phi_3\left(\sum_{1 \leq k \leq t_3} \lambda_k \gamma_3^A(k) \right)}  = \displaystyle{\sum_{1 \leq k \leq t_3} \lambda_k \gamma_3(k)}$$
It means that $X = X_{1} \cup X_{2}$ admits the signature $\lambda$ in $T$, since $\gamma_3$ is the boundary of $N$.
\end{proof}

\begin{lemma}
\label{unbuildbis}
Let $X_{1}$ and $X_{2}$ be two sets of internal elements of $N_1$ and $N_2$.
If $X = X_{1} \cup X_{2}$ admits $\lambda$ in $T$, then there are two signatures $\mu$ and $\delta$
such that $R(A, \lambda,\mu,\delta)$ holds, $X_1$ admits $\mu$ in $T_1$ and $X_{2}$ admits $\delta$ in $T_{2}$.
\end{lemma}
\begin{proof}
 Since $X$ admits $\lambda$ in $T$, there is a linear combination of elements of $X$ equal to $\displaystyle{\phi_3\left(\sum_{1 \leq k \leq t_3} \lambda_k \gamma_3(k) \right)}$. It is equivalent to say that we have the following equality in $V$:
 \begin{equation}
  \label{eqter}
 (v_1,0,0) + (0,v_2,0) + (b_1,0,b_2) + (0,c_1,c_2) = \displaystyle{\sum_{1 \leq k \leq t_3} (0,0,\lambda_k \gamma_3^A(k))},
 \end{equation}
 where $ (v_1,0,0)$ is a combination of elements of $X_1$ injected in $V$, $(0,v_2,0)$ is a combination of elements of $X_2$ injected in $V$, $ (b_1,0,b_2) \in B$ and $(0,c_1,c_2) \in C$.
Since $ (b_1,0,b_2)$ is in $B$, there is a signature $\mu$ such that it is equal to: 
 $$  \displaystyle{\sum_{1 \leq i \leq t_1}( \mu_i \gamma_1(i), 0, -\mu_i \gamma_1^{A}(i)})$$
 In the same way, there is a signature $\delta$ such that $ (0,c_1,c_2)$ is equal to: 
  $$  \displaystyle{\sum_{1 \leq j \leq t_2}(0, \delta_j \gamma_1(j), - \delta_i \gamma_1^{A}(j)})$$
Equation \ref{eqter} implies that $v_1=-b_1$ and $v_2= -c_1$, therefore $X_1$ is of signature $\mu$ in $T_1$ and 
$X_2$ is of signature $\delta$ in $T_2$. We also deduce from Equation \ref{eqter}: $$b_2+ c_2 = \displaystyle{\sum_{1 \leq k \leq t_3} (0,0,\lambda_k \gamma_3^A(k))}$$
Therefore $R(A,\lambda,\mu,\delta)$ holds.
\end{proof}

By means of these two lemmas, we can prove that enhanced trees of width $t$ and  $\mathcal{M}_{t}^{\mathbb{F}}$ are the same object.
Let $g$ be the after defined bijection between the enhanced trees of width $t$ and $\mathcal{M}_{t}^{\mathbb{F}}$.
Let $T$ be an enhanced tree, one replaces $A$ on each internal node by  $\odot_{A} $ (a characteristic matrix is a $3$-partitioned matrix).
The images of the leaves labeled $(0)$ and $(1)$ are the constants $\varUpsilon_{0}$ and $\varUpsilon_{1}$ respectively.

\begin{theorem}
\label{equi}
Let $M$ be a $\mathbb{F}$-matroid, then 
$T$ is one of its enhanced tree of width $t$  if and only if $M$ is the value of the term $g(T)$.
\end{theorem}

\begin{proof}
One can prove a theorem of characterization of dependent sets by the signatures on the terms of $ \mathcal{M}_{t}^{\mathbb{F}}$
identical to Theorem \ref{finale}, using Lemmas \ref{buildbis} and \ref{unbuildbis}. Therefore $T$ and $g(T)$ define the same matroid. 
\end{proof}

\begin{example}

 \begin{figure}[ht]
 \centering
 \includegraphics[scale=0.4]{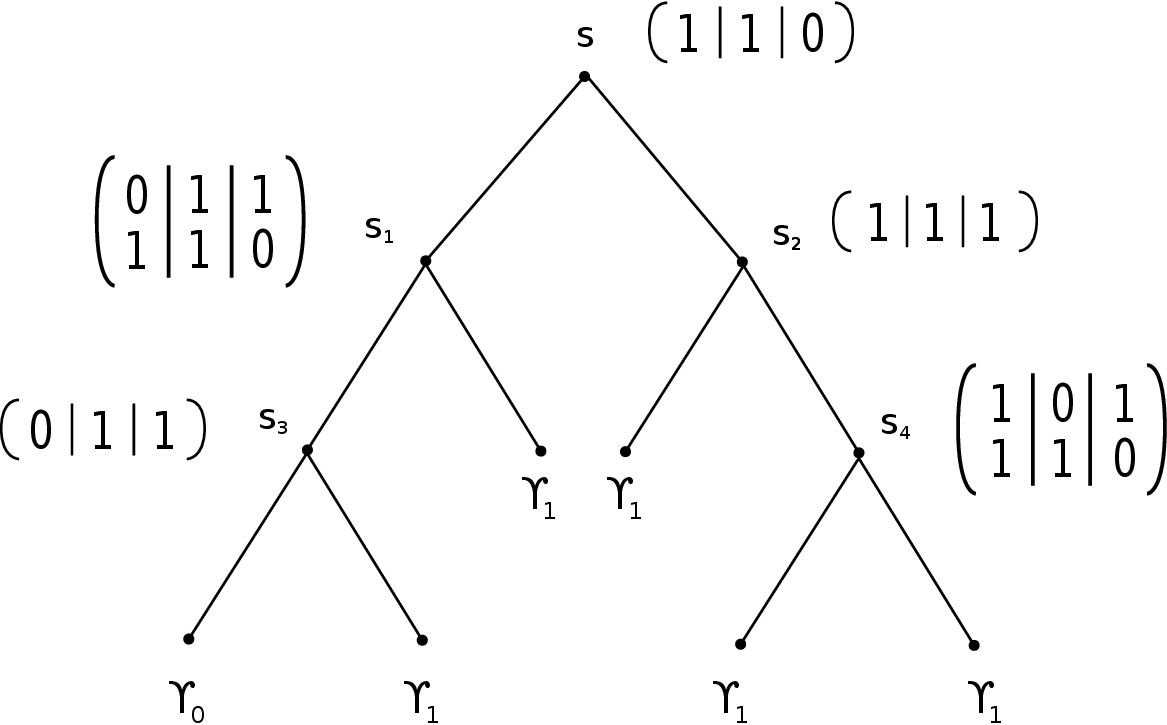}
\caption{The term associated to the enhanced tree of Fig. \ref{fig:enhanced} }
 \label{fig:term}
\end{figure}

We give here the matrices, with their boundary on the left side, which are constructed 
when evaluating the term of Fig. \ref{fig:term}
$$
\begin{array}{cc}
\begin{array}{c}
 M_{s_{3}} = \left(\begin{array}{c|cc}
0&1&0  \\
1&0&1  
\end{array}\right) \\
M_{s_{1}} = 
\left( \begin{array}{c|ccc}
0 & 1 & 0 & 0 \\
0 & 0 & 1 & 1 \\
1 & 0 & 0 & 1 \\
\end{array} \right)
\end{array}
 &
\begin{array}{c}
 M_{s_{4}} = \left(\begin{array}{c|cc}
1 & 1 & 0  \\
0 & 1 & 1  
\end{array}\right)
 \\
M_{s_{2}} =
\left( \begin{array}{c|ccc}
1 & 1 & 0 & 1  \\
0 & 1 & 1 & 0  \\
\end{array} \right)
\end{array}
\end{array}
$$

$$
M_{s} = 
\left( \begin{array}{cccccc}
1 & 0 & 0 & 0 & 0 & 0\\
0 & 1 & 1 & 0 & 0 & 0\\
0 & 0 & 1 & 1 & 0 & 1\\
0 & 0 & 0 & 1 & 1 & 0
\end{array} \right)
$$

The matrix $M_{s}$ represents the same matroid as the matrix $X$ of Fig. \ref{fig:tree} which was used to find an enhanced tree and then a term as explained in the
proof of the previous theorem.
\end{example}


\subsection{Series and parallel connections}

In this subsection we consider two of the most simple operations on matroids, called the series and parallel connections. 
They extend well-known graph operations, which are used to characterize the graphs of tree-width $2$ \cite{bodlaender1998pka}.
By means of these operations, we describe a class of matroids, which are not all representable, using the methods introduced in the previous subsection.
The following definition and theorem are taken from \cite{oxley1992mt}.

\begin{definition}
 Let $M_{1}$ and $M_{2}$ be two $1$ boundaried matroids of ground set $S_{1}$ and $S_{2}$.
Their respective boundaries are $\left\lbrace p_{1}\right\rbrace $ and $\left\lbrace p_{2}\right\rbrace $.
 We denote by $\mathcal{C}(M)$ the collection of circuits of the matroid $M$.
Let $E$ be the set $S_{1} \cup S_{2} \cup \left\lbrace p \right\rbrace \setminus\left\lbrace p_{1}, p_{2} \right\rbrace $.
We define two collections of subsets of $E$:

$$\begin{array}{l} C_{S} = \left\lbrace \begin{array}{l l}
 \mathcal{C}(M_{1}\setminus \left\lbrace p_{1} \right\rbrace)  \cup  \mathcal{C}(M_{2}\setminus \left\lbrace p_{2} \right\rbrace)  \\
 \cup \left\lbrace C_{1}\setminus \left\lbrace p_{1} \right\rbrace \cup  C_{2}\setminus \left\lbrace p_{2}\right\rbrace \cup \left\lbrace p \right\rbrace \,|\, p_{i} \in C_{i} \in \mathcal{C}(M_{i})  \right\rbrace
  \end{array}\right. 
\\
~
\\
C_{P} =  \left\lbrace  \begin{array}{l l}
 \mathcal{C}(M_{1}\setminus \left\lbrace p_{1} \right\rbrace)  \cup  \mathcal{C}(M_{2}\setminus \left\lbrace p_{2} \right\rbrace) \\
 \cup_{i=1,2} \,\, \left\lbrace C_{i}\setminus \left\lbrace p_{i} \right\rbrace \cup \left\lbrace p \right\rbrace \,|\, p_{i} \in C_{i} \in \mathcal{C}(M_{i})  \right\rbrace \\
  \cup \left\lbrace C_{1}\setminus \left\lbrace p_{1} \right\rbrace \cup  C_{2}\setminus \left\lbrace p_{2}\right\rbrace \,|\, p_{i} \in C_{i} \in \mathcal{C}(M_{i})  \right\rbrace 
            \end{array}\right. 
\end{array}
  $$

\end{definition}

\begin{theorem}
 The sets $C_{S}$ and $C_{P}$ are collections of circuits of a matroid on $E$.
\end{theorem}

\begin{figure}[ht]
 \centering
 \includegraphics[scale=0.3]{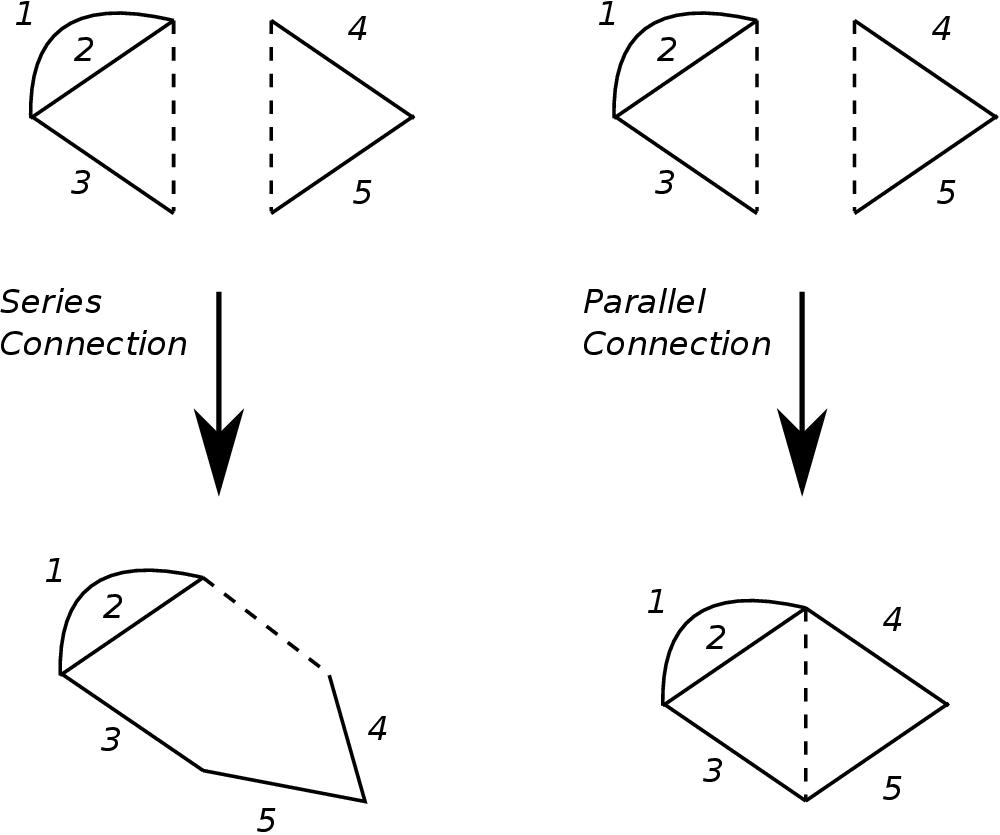}
\caption{Example of series and parallel connections over graphs with boundaries represented by a dotted line}
\end{figure}  

The matroid defined by $C_{P}$ is called the parallel connection of $M_1$ and $M_2$
while the one defined by $C_{S}$ is the series connection of $M_1$ and $M_2$. 

\begin{definition}
We write $M_1 \oplus_{p} M_2$ for the parallel connection of $M_1$ and $M_2$ restricted 
to the ground set $E \setminus\left\lbrace p \right\rbrace$.
\end{definition}

The operator $\oplus_{p}$ is known under the name of \emph{2 sum} (see \cite{oxley1992mt}).
We could also consider an operator $\oplus_{s}$, but it is only the direct sum of two matroids
and it will not enlarge the class of matroids we are about to define.
We now consider the operation $\odot$ defined as in Definition \ref{def:odot},
except that $\oplus$ is replaced by $\oplus_{p}$.

\begin{definition}
Let $\mathcal{L}_{k}$ be the set of $1$ boundaried matroids of size at most $k$ and let $\mathcal{M}$ be
the set of $3$-partitioned matroids of size $3$. We write $\mathcal{T}_k$ for the set of terms $T(\mathcal{L}_{k},\mathcal{M})$.
\end{definition}

A term of  $\mathcal{T}_k$ has for value a $1$ boundaried matroid.
Remark that there are only $6$ different matroids of size $3$ up to isomorphism (see the proof of Lemma \ref{relation_simple}).
Notice also that the class of boundaried matroids of size $k$ closed by the series parallel operation is strictly larger than $\mathcal{T}_k$.
Indeed, when one builds a term, the position of the boundary is imposed.
It could be interesting to extend the result of this section to this broader class.

A term of $\mathcal{T}_k$ can have a non representable matroid for value, since the constants at the leaves are arbitrary matroids. 
Therefore the matroids represented by elements of $\mathcal{T}_k$ and elements of $\mathcal{M}_t^{\mathbb{F}}$ are different.
Nevertheless there is a relation between the operations $\oplus_{p}$ and $\oplus$ as illustrated by the next proposition.

\begin{proposition}
 Let $M_{1}$ (resp. $M_2$) be a matroid of boundary $\left\lbrace p_{1}\right\rbrace $ (resp. $\left\lbrace p_{2}\right\rbrace $) 
represented by the sets of vectors $A_{1}$ (resp. $A_2$).
Then $A_{1} \oplus A_{2}$ represents the matroid $M_{1} \oplus_{p} M_{2}$. 
\end{proposition}

\begin{proof}
We prove that the dependent sets of $A_{1} \oplus A_{2}$ are the same as the dependent sets of $M_{1} \oplus_{p} M_{2}$.
In fact, we only show that a dependent set $D$ of $M_{1} \oplus_{p} M_{2}$ is a dependent set of $A_{1} \oplus A_{2}$.
The converse is easy and left to the reader.
By the definition of $\oplus_{p}$, the dependent set $D$ can be of two different kinds.
It may be the image of a dependent set of $M_{1} \setminus \{p_1\}$ or $M_{2} \setminus \{p_2\}$,
it is then trivially a dependent set of $A_{1} \oplus A_{2}$.

Assume now that $D = D_{1}\setminus \{p_1\} \cup D_{2}\setminus \{p_2\}$, where $D_1$ is a dependent set of 
$M_1$ containing $\{p_1\}$ and $D_2$ a dependent set of $M_2$ containing $\{p_2\}$.
Since $M_1$ and $M_2$ are represented by $A_1$ and $A_2$, we have the following linear 
 dependence relations of their columns in bijection with $D_1$ and $D_2$:
 $$\displaystyle{\lambda_{1} p_{1} + \sum \alpha_{i} A_{1}^{i} = 0} \text{ and } \displaystyle{\lambda_{2} p_{2} + \sum \beta_{i} A_{2}^{i} = 0}$$ 

By linear combination of the two previous equalities we get: 
$$ \lambda_{1}(p_{1} - p_{2}) + \displaystyle{ \sum \alpha_{i} A_{1}^{i} +\sum -\lambda_{1}\lambda_{2}^{-1}\beta_{i} A_{2}^{i} = 0}$$ 
In $A_{1} \oplus A_{2}$, we have $p_1 = p_2$ therefore, the equation becomes:
$$ \displaystyle{ \sum \alpha_{i} A_{1}^{i} +\sum -\lambda_{1}\lambda_{2}^{-1}\beta_{i} A_{2}^{i} = 0}$$
This last equation proves that $D$ is dependent in  $A_{1} \oplus A_{2}$.
\end{proof}

It seems that the previous lemma would fail for generalizations of $\oplus_{p}$ to a boundary larger than one.
Indeed, $\oplus$ is not an operation on matroids as seen in Example \ref{counterexample} with a boundary of size two. 
In fact, one of the natural generalizations of $\oplus_{p}$ to boundary of size $k$ is the $k$ sum (see \cite{oxley1992mt})
which is defined on binary matroids only.

\begin{corollary}
A matroid defined by a term of $\mathcal{T}_k$ whose constants are $\mathbb{F}$-matroids is 
an $\mathbb{F}$-matroid of branch-width at most $k$.
\end{corollary}
 
\begin{proof}
By structural induction on the terms of $\mathcal{T}_k$ whose constants are $\mathbb{F}$-matroids.
The constants are matroids of size $k$ because they are in $\mathcal{T}_k$ and they are representable by hypothesis,
 hence they are of branch-width at most $k$.
Assume now that $T= T_1 \odot_M T_2$, where the values of $T_1$ and $T_2$ are matroid of branch-width $k$
represented by $A_1$ and $A_2$ respectively. All matroids of size $3$ are cycle matroids and hence are representable in any field.
Therefore $M$ is represented by the $3$-partitioned matrix $A$ over $\mathbb{F}$.
Using the previous proposition, we have that $A_1 \odot_A A_2$ represents the same boundaried matroid as $T= T_1 \odot_M T_2$.
Finally, Theorem \ref{equi} proves that $A_1 \odot_A A_2$ is of branch-width $k$, which completes the proof.
\end{proof}

We now define a very general notion of signature to use the previously introduced technique
and illustrate it in this setting. 
A signature describes which sets of elements of the boundary make a set of internal elements dependent. 
Notice that, contrary to the representable matroid case, the signature of a set is unique.
We could use this notion of signature for other operations than $\oplus_p$,
over matroids of boundary bigger than $1$.

\begin{definition}[Signature]
 Let $T$ be a term whose value is a boundaried matroid $M$ and let $X$ be a set of internal elements of $M$. The signature 
of the set $X$ in $T$ is the set of all the subsets $S$ of the boundary such 
that $X \cup S$ is a dependent set in $M$. 
\end{definition}


In general, if the boundary is of size $s$, there are less than $2^{2^s}$ signatures.
In our setting, the term is in $\mathcal{T}_k$, thus there is only one boundary element denoted by $1$.
We have only three different signatures:
\begin{enumerate}
\item if $X$ is dependent then it is of signature $\{\{\},\{1\}\}$ that we denote by $\mathbf{2}$
\item if $X$ is dependent only when we add the boundary element then it is of signature $\{\{1\}\}$ which we denote by $\mathbf{1}$
\item if $X$ is independent even with the boundary element then it is of signature $\emptyset$ which we denote by $\mathbf{0}$
\end{enumerate}

Note that an empty set is of signature $\mathbf{0}$, because the boundary is an independent set.
We now prove in this context a result similar to Lemma \ref{build}.
We will not need an equivalent of Lemma \ref{unbuild}, since here the signatures are unique.

\begin{lemma}
\label{relation_simple}
There is a relation $R_{p}(\mu,\delta,\lambda, N)$, where the first three arguments are signatures and $N$ 
is a $3$-partitioned matroid of size $3$, such that the following holds.
Let $T = T_1 \odot_N T_2$ be a term of $\mathcal{T}_k$, let $X_{1}$ and $X_2$ be sets of internal elements
of the boundaried matroids represented by respectively $T_1$ and $T_2$.
If the set $X_{1}$ is of signature $\mu$ in $T_1$, the set $X_2$ is of signature $\delta$ in $T_2$
and $R_{p}(\mu,\delta,\lambda, N)$ holds, then the set $X_{1} \cup X_{2}$ is of signature $\lambda$ in $T$.
\end{lemma}
\begin{proof}

There are six $3$-partitioned matroids of size $3$, which we denote by $N_i$ for $i=1,\dots,6$.
We represent each of them by three points in an ellipse.
The bottom left point is $\gamma_1^{N_i}(1)$, the bottom right one is $\gamma_2^{N_i}(1)$ and the top one is $\gamma_3^{N_i}(1)$.
The smaller ellipses enclosing points represent the circuits of the matroid. 
 We give here the value of the relation $R_{p}$ for each $N_i$.
One may then easily check that the proposition holds.

\begin{center}
\begin{tabular}{lll}
$N_1$ & \parbox[c]{2cm}{\includegraphics[scale=0.4]{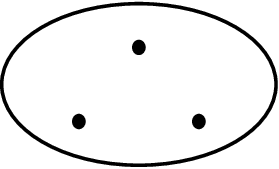}} & \begin{tabular}{l}
$R_p(\cdot,\cdot,\mathbf{2},N_1) = \{ (\mathbf{0},\mathbf{2}),(\mathbf{2},\mathbf{0}),(\mathbf{1},\mathbf{2}),(\mathbf{2},\mathbf{1}),(\mathbf{2},\mathbf{2})\}$\\
$R_p(\cdot,\cdot,\mathbf{1},N_1) = \{ \}$
\end{tabular}\\\\
$N_2$ & \parbox[c]{2cm}{\includegraphics[scale=0.4]{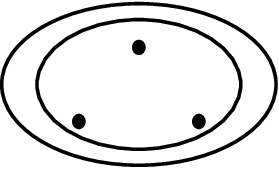}}& 
\begin{tabular}{l}
$R_p(\cdot,\cdot,\mathbf{2},N_2) = \{ (\mathbf{0},\mathbf{2}),(\mathbf{2},\mathbf{0}),(\mathbf{1},\mathbf{2}),(\mathbf{2},\mathbf{1}),(\mathbf{2},\mathbf{2})\}$\\
$R_p(\cdot,\cdot,\mathbf{1},N_2) = \{ (\mathbf{1},\mathbf{1})\}$
\end{tabular}\\\\
$N_3$ & \parbox[c]{2cm}{\includegraphics[scale=0.4]{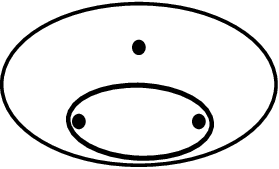}} & 
\begin{tabular}{l}
$R_p(\cdot,\cdot,\mathbf{2},N_3) = \{ (\mathbf{0},\mathbf{2}),(\mathbf{2},\mathbf{0}),(\mathbf{1},\mathbf{2}),(\mathbf{2},\mathbf{1}),(\mathbf{2},\mathbf{2}),(\mathbf{1},\mathbf{1})\}$\\
$R_p(\cdot,\cdot,\mathbf{1},N_3) = \{ \}$
\end{tabular}\\\\
$N_4$ & \parbox[c]{2cm}{\includegraphics[scale=0.4]{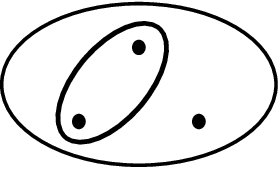}} & 
\begin{tabular}{l}
$R_p(\cdot,\cdot,\mathbf{2},N_4) = \{ (\mathbf{0},\mathbf{2}),(\mathbf{2},\mathbf{0}),(\mathbf{1},\mathbf{2}),(\mathbf{2},\mathbf{1}),(\mathbf{2},\mathbf{2})\}$\\
$R_p(\cdot,\cdot,\mathbf{1},N_4) = \{ (\mathbf{1},\mathbf{1}),(\mathbf{1},\mathbf{0})\}$
\end{tabular}\\\\
$N_5$ & \parbox[c]{2cm}{\includegraphics[scale=0.4]{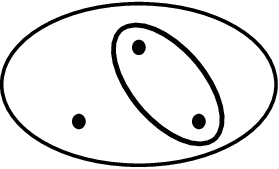}} & 
\begin{tabular}{l}
$R_p(\cdot,\cdot,\mathbf{2},N_5) = \{ (\mathbf{0},\mathbf{2}),(\mathbf{2},\mathbf{0}),(\mathbf{1},\mathbf{2}),(\mathbf{2},\mathbf{1}),(\mathbf{2},\mathbf{2})\}$\\
$R_p(\cdot,\cdot,\mathbf{1},N_5) = \{ (\mathbf{1},\mathbf{1}),(\mathbf{0},\mathbf{1})\}$
\end{tabular}\\\\
$N_6$ & \parbox[c]{2cm}{\includegraphics[scale=0.4]{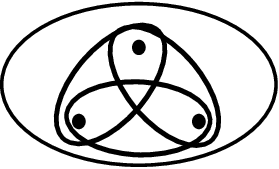}}&
\begin{tabular}{l}
$R_p(\cdot,\cdot,\mathbf{2},N_5) = \{ (\mathbf{0},\mathbf{2}),(\mathbf{2},\mathbf{0}),(\mathbf{1},\mathbf{2}),(\mathbf{2},\mathbf{1}),(\mathbf{2},\mathbf{2}),(\mathbf{1},\mathbf{1})\}$\\
$R_p(\cdot,\cdot,\mathbf{1},N_5) = \{ (\mathbf{0},\mathbf{1}),(\mathbf{1},\mathbf{0})\}$
\end{tabular}
\end{tabular}
\end{center}
\end{proof}
 
We give in the proof the value of the relation $R_{p}$ which plays the same role as $R$ in Lemma \ref{build}.
The precise value of $R_p$ is not important for the proof: what matters is that it only depends on $\mu$, $\delta$, $\lambda$ 
and $N$, but not on $X_{1}$, $X_{2}$ or $T$.

A close examination of the operations $\odot_{N_{i}}$ in the previous proof
shows that we already know three of them:
\begin{itemize}
\item $M_{1} \odot_{N_{1}} M_{2}$ is the matroid given by the direct sum of $M_{1}$ and $M_{2}$ with boundary $\gamma_3^{N_1}$.
 \item $M_{1} \odot_{N_{2}} M_{2}$ is the matroid given by the series connection of $M_{1}$ and $M_{2}$ with boundary $\gamma_3^{N_2}$.
\item $M_{1} \odot_{N_{6}} M_{2}$ is the matroid given by the parallel connection of $M_{1}$ and $M_{2}$ with boundary $\gamma_3^{N_6}$.
\end{itemize}

Observe that a leaf of a term of $\mathcal{T}_k$ represents a matroid of size less than $k$,
while a leaf of a term in $\mathcal{M}_t^{\mathbb{F}}$ represents one element of the matroid it defines.
To use our method on terms of $\mathcal{T}_k$, it is convenient to modify them.
At each leaf labeled by an abstract $1$ boundaried matroid $M$, we root a binary tree with as many leaves as internal elements in $M$. 
We denote by $\tilde{\mathcal{T}_k}$, the sets of terms of $\mathcal{T}_k$ transformed in this way.
We now have a bijection between the leaves of a term of $\tilde{\mathcal{T}_k}$ and the elements of the matroid it represents.
 
\begin{theorem}[Characterization of dependence]
 Let $T$ be a term of $\tilde{\mathcal{T}_k}$ which represents the matroid $M$ and let $X$ be a set of elements of $M$.
  The set $X$ is dependent if and only if there exists a signature $\lambda_{s}$ at each node $s$ of $T$ seen as a labeled tree:
  \begin{enumerate}
     \item  if $s_{1}$ and $s_{2}$ are the children of $s$ of label $\odot_{N}$ then $R_{p}(\lambda_{s_{1}},\lambda_{s_{2}},\lambda_{s},N)$
     \item if $s$ is labeled by an abstract boundaried matroid $N$, then $X \cap N$ is a set of signature $\lambda_{s}$ in $N$
   \item the signature at the root is $\mathbf{2}$
  \end{enumerate}
\end{theorem}
\begin{proof}
 Let us remark that the set $X$ is dependent in $M$ if its signature contains the set $\{\}$,
i.e. if it is $\mathbf{2}$.
We thus have to prove by induction on $T$ that $\lambda_s$ is the signature of $X$ in $T_s$.
The base case is given by the condition $2$, while Lemma \ref{relation_simple} and 
condition $1$ allow us to prove the induction step.
\end{proof}

The function $F$ we use in the next theorem is the same as in Section \ref{sec:rep_matroid}.
It associates a formula of $MSO$ to a formula of $MSO_{M}$ by relativization to the leaves 
 and the use of a formula $dep$, whose new definition is given in the proof of the next theorem.

\begin{theorem}
\label{th:parallel}
There exists a mapping $F$ such that if $T$ is a term of $\tilde{\mathcal{T}_k}$ which represents the matroid $M$ 
and if $f$ is the bijection between the leaves of $T$ and the elements of $M$ then $M \models \phi(\vec{a}) \Leftrightarrow T\models F(\phi(f(\vec{a})))$.
\end{theorem}

\begin{proof}
The demonstration is done by the construction of a formula $dep(Y)$ satisfying 
the conditions of the characterization theorem. We use the formulas defined in the proof of Theorem \ref{main},
condition $1$ is implemented by the formula $\Psi_{1}$ except that $R$ is now the relation $R_{p}$.
In $\Psi_{3}$, we replace $X_{(0,\dots,0)}$ by $X_{\mathbf{2}}$ to satisfy condition $3$.

Let $Q(S,N,\lambda)$ be the relation which is true if and only if $S$ is a subset of the boundaried matroid $N$
of signature $\lambda$. 
Recall that the set of signatures $\lambda_s$ is represented by a set of second-order variables $\vec{X}$.
To enforce condition $2$, we define a formula $\Psi_4(X,\vec{X},s)$.
It is true if and only if each internal node $s$ of signature $\lambda$ is labeled by a boundaried matroid $N$ 
and $\lambda$ is indeed the signature of the intersection of $Y$ with $N$.
We write $Y \cap N = S$ for the fact that the elements of $Y$ which are leaves of a subtree rooted
in a node labeled by the boundaried matroid $N$ form the subset $S$ of $N$. One may check that it is expressible
by a $MSO$ formula. 

$$ \Psi_2(Y,\vec{X},s) = \displaystyle{\bigwedge_{(N, S \subseteq N),\lambda} (label(s) = N \wedge X_{\lambda}(s) \wedge Y \cap N = S) \Rightarrow Q(S,N,\lambda) }   $$

This formula is a conjunction on all boundaried matroids $N$ of size $k$ and their subsets,
which are in number bounded by $2^{2^{k}}$, and on the three possible signatures.
We define the formula $dep$ of size  $O(2^{2^{k}})$:
$$dep(Y) \equiv  \exists \vec{X}  \,\Omega(\vec{X}) \wedge \Psi_{1}( \vec{X} ) \wedge \Psi_{2}(Y, \vec{X})\wedge \Psi_{3}(\vec{X}) $$

The characterization theorem proves that the formula $dep$ is correct
and the theorem is then obtained by a simple induction on the formula.
\end{proof}

\begin{corollary}
 The model-checking problem of $MSO_{M}$ over the set of matroids given by a term of $\mathcal{T}_k$ is decidable in time $f(k,l) \times n$,
where $n$ is the number of elements of the matroid, $l$ is the size of the formula and $f$ is a computable function.
\end{corollary}

\section{Conclusion}


In this article, we have studied the representable matroids of bounded branch-width.
We have given a new proof of the fact that model-checking of $MSO_M$ over them can be done in polynomial time (linear if 
a suitable representation is given). Moreover we have linked together the notion of enhanced tree, adapted from 
the branch decomposition, and the terms of $\mathcal{M}_t^{\mathbb{F}}$.
In both cases, we use the same tools, namely the relation $R$ and the characterization of dependent sets through
signatures and $R$.

We have also introduced the set of terms $\mathcal{T}_k$, which represent matroids different from $\mathcal{M}_t^{\mathbb{F}}$.
We have then used the exact same method with signatures and a relation $R_p$ to characterize the dependent sets in 
a matroid represented by a term of $\mathcal{M}_t^{\mathbb{F}}$. 
In fact, we could use this method on any term built from an operation $M_1 \oplus M_2$, such that 
$M_1$ and $M_2$ are restrictions of $M_1 \oplus M_2$. In other words, the operation has to be derived from an amalgam
(or push-out) over a class of matroids.

One natural generalization to our construction, would be to
 lift the condition that the boundaries are independent sets
and thus build more terms from $\oplus$. But it does not seem
that we can obtain more matroids in this way.
In the other hand, if we want to extend the operation of
parallel connection to a boundary of any size, the properties 
of the boundary play a big role.
There are thus two natural open questions: 
\begin{itemize}
 \item How to generalize the class $\mathcal{T}_k$ by allowing boundaries of size larger than one? 
\item Is it possible to design a matroid grammar which unifies both classes
presented in this paper (and possibly more)?
\end{itemize}






\bibliographystyle{elsarticle-num.bst} 
\bibliography{ma_biblio.bib}

\end{document}